\documentclass[onecolumn, draftclsnofoot,11pt]{IEEEtran}

\usepackage[utf8]{inputenc} 
\usepackage[T1]{fontenc}
\usepackage{url}
\usepackage{ifthen}
\usepackage{cite}
\usepackage[cmex10]{amsmath} 

\newtheorem{theorem}{Theorem}
\newtheorem{lemma}[theorem]{Lemma}
\newtheorem{prop}{Proposition}

\newtheorem{proposition}[theorem]{Proposition}
\newtheorem{definition}[theorem]{Definition}
\newenvironment{proof}{{ \bf \emph{Proof.} }}{\hfill $\Box$ \\} 

\global\long\def\RR{\mathbb{R}}

\global\long\def\EE{\mathbb{E}}
\global\long\def\PP{\mathbb{P}}

\newcommand{\bfb}{\mathbf{b}}

\newcommand{\bfA}{\mathbf{A}}

\def \sfb {\mathsf{b}}
\def \sfxn {\mathsf{x}^n}
\def \sfL {\mathsf{L}}

\newcommand{\sfA}{\mathsf{A}}
\newcommand{\sfB}{\mathsf{B}}
\newcommand{\sfS}{\mathsf{S}}
\newcommand{\sfU}{\mathsf{U}}
\newcommand{\sfX}{\mathsf{X}}
\newcommand{\sfY}{\mathsf{Y}}
\newcommand{\sfZ}{\mathsf{Z}}

\def \sfXhat {\hat{\sfX}}

\def \tensor {\otimes}

\def \sigmabar {\bar{\sigma}}

\def\<{\langle}
\def\>{\rangle}

\def \calA {\mathcal{A}}
\def \calB {\mathcal{B}}
\def \calC {\mathcal{C}}
\def \calD {\mathcal{D}}
\def \calE {\mathcal{E}}

\def \calH {\mathcal{H}}

\def \calL {\mathcal{L}}
\def \calM {\mathcal{M}}

\def \calP {\mathcal{P}}

\def \calR {\mathcal{R}}

\def \calU {\mathcal{U}}

\def \calW {\mathcal{W}}
\def \calX {\mathcal{X}}

\def \calZ {\mathcal{Z}}

\def \Rbar {\bar{R}}
\def \Pibar {\bar{\Pi}}
\def \btheta {\bar{\theta}}
\def \bTheta {\bar{\Theta}}

\newcommand{\Tr}{\text{Tr}}
\usepackage{stackengine}
\def \deq {:=}
\def \textT {\text{T}}
\def \eqand {\text{ and }}

\def\rhotilde{\tilde{\rho}}
\def\omegabar{\bar{\omega}}

\def \I {\mathds{1}}

\newcommand{\curly}[1]{\left\{{#1}\right\}}
\newcommand{\round}[1]{\left({#1}\right)}

\def \an{a^n}
\def \abarn{\bar{a}^n}
\def \An{A^n}

\def \Bn{B^n}

\def \un{u^n}
\def \Un{U^n}
\def \xn{x^n}
\def \Xn{X^n}
\def \yn{y^n}
\def \Yn{Y^n}
\def \zn{z^n}
\def \Zn{Z^n}

\def \codeDistribution {\mathbb{P}}
\def \Txqc {\mathcal{T}_{{\delta}}^{(n)}({X})}
\def \refe {R}
\def \epovm {\Omega}
\def \dpovm {\Lambda}
\def \utilde {\Tilde{u}}
\def \ltilde {\Tilde{l}}
\def \Ipmf {\I_{\{\mbox{sPMF}\}}}
\newcommand{\set}[1]{\{#1\}}
\def \targetpu{P_U}
\def \Tuqc {\mathcal{T}_{{\delta}}^{(n)}({U})}
\def \Tu{\mathcal{T}_\delta^{(n)}(U)}
\def \Tucond{\mathcal{T}_\delta^{(n)}(U|x^n)}
\def \Tuz{\mathcal{T}_\delta^{(n)}(U,Z)}
\def \Txcond{\mathcal{T}_\delta^{(n)}(X|u^n)}
\def \bTheta {\bar{\Theta}}

\def \cpe {\Delta_{\text{CP}}}
\def \ce {\Delta_{\text{C}}}
\def \notce {\Tilde{\Delta}_{\text{C}}}
\def \ee {{\Delta}_{\text{E}}}
\def \epe {{\Delta}_{\text{EP}}}
\def \er {{\Delta}}

\def \peone {\Delta_{\text{P}}^{(1)}}
\def \peoneone {\Delta_{\text{P}}^{(11)}}
\def \petwo {\Delta_{\text{P}}^{(2)}}
\def \peonetwo {\Delta_{\text{P}}^{(12)}}

\def \xhat {\hat{x}}
\def \pu {{P}_U}
\def \px {{P}_X}
\def \pxz {{P}_{XZ}}

\def \prevTC {{W}_{X|\hat{X}}}
\def \Tx {\mathcal{T}_{\hat{\delta}}^{(n)}({X})}

\def \codebook{\calC}
\def \encodern {\mathcal{E}^{(n)}}
\def \decodern {\mathcal{D}^{(n)}}
\def \codeDistribution {\mathbb{P}}
\def \error {\Xi(\Gamma^{(n)}_\sfA,\Gamma^{(n)}_\sfB,f)}

\usepackage{bbm}
\usepackage[utf8]{inputenc} 
\usepackage[T1]{fontenc}
\usepackage{url}              
\usepackage{cite}             

\usepackage[cmex10]{amsmath}  
\usepackage{mleftright}       
\mleftright                   

\usepackage{graphicx}         
\usepackage{booktabs}         

\usepackage{dsfont}
\usepackage{amssymb}
\usepackage{graphicx}
\usepackage[colorlinks=true, allcolors=red]{hyperref}
\usepackage{cite}
\begin{document}
\title{When Wyner and Ziv Met Bayes in Quantum-Classical Realm} 


\author{\IEEEauthorblockN{Mohammad Aamir Sohail\IEEEauthorrefmark{1}, Touheed Anwar Atif\IEEEauthorrefmark{2}, and
		S. Sandeep Pradhan\IEEEauthorrefmark{1}
		\\}
	\IEEEauthorblockA
		\IEEEauthorrefmark{1}{Department of EECS, University of Michigan, Ann Arbor, USA.\\
		\IEEEauthorrefmark{2}Los Alamos National Laboratory, USA\\
		Email: \tt mdaamir@umich.edu, tatifd@lanl.gov,  pradhanv@umich.edu
		}
}

\maketitle


\begin{abstract}
   In this work, we address the lossy quantum-classical source coding with the quantum side-information (QC-QSI) problem. The task is to compress the classical information about a quantum source, obtained after performing a measurement while incurring a bounded reconstruction error. Here, the decoder is allowed to use the side information to recover the classical data obtained from measurements on the source states.
   We introduce a new formulation based on a backward (posterior) channel, replacing the single-letter distortion observable with a single-letter posterior channel to capture reconstruction error. Unlike the rate-distortion framework, this formulation imposes a block error constraint. An analogous formulation is developed for lossy classical source coding with classical side information (C-CSI) problem. We derive an inner bound on the asymptotic performance limit in terms of single-letter quantum and classical mutual information quantities of the given posterior channel for QC-QSI and C-CSI cases, respectively. Furthermore, we establish a connection between rate-distortion and rate-channel theory, showing that a rate-channel compression protocol attains the optimal rate-distortion function for a specific distortion measure and level.
\end{abstract}

\section{Introduction}
\label{sec:intro}

Reliable compression of an information source for efficient storage is a central problem in information theory. A key question involves the asymptotic characterization of the rate required to compress a source that can be recovered up to a certain degree of measured loss, which is known as the lossy source coding problem. In the case of classical sources, Shannon’s rate-distortion theory (RDT) \cite{shannon1959coding} established that the optimal rate for lossy data compression for a memoryless source is given by the mutual information between the source and its reconstruction. In this theory, Shannon introduced a \textit{local} error criterion based on averaged symbol-wise distortion between the reconstruction and the original source, using metrics like the Hamming distance and mean squared error.


The Wyner-Ziv (WZ) \cite{wyner1976rate} theorem extends Shannon's lossy source coding framework to a network setting where correlated side information about the source is available only at the decoder. 
It characterizes distortion using the averaged local error criterion while minimizing the required rate by utilizing side information to reconstruct the source.  
The WS rate-distortion function is characterized by the mutual information between the source and an auxiliary random variable conditioned on the side information. 

Considering side information in the quantum realm,  
Devetak and Winter \cite{devetak2003classical} explored lossless classical compression with quantum side information (QSI), demonstrating that quantum correlations can reduce the required classical compression rate. 
In the domain of lossy classical compression with QSI, Luo et al. \cite{luo2009channel} characterized the rate region when classical data is compressed in the presence of QSI at the decoder. 
In the context of the quantum-classical (QC) case, Datta et al. \cite{datta2013quantum} explored QC compression using QSI. The authors derived an asymptotic QC rate-distortion function in terms of single-letter quantum mutual information conditioned on the quantum side information and characterized distortion through a local error criterion, which is a QC analogue to those found in Shannon's and WZ's lossy source coding problems.

In our previous work \cite{sohail2023unique}, we developed a new formulation of the lossy source coding problem, called \textit{rate-channel theory} (RCT), which departs from the local error criterion to characterize distortion and employs a \textit{global} error criterion based on the posterior (backward) channel. In RCT, a single-letter posterior channel is given that characterizes the nature of loss instead of a single-letter distortion function. The goal is to construct a coding system such that the joint effect of producing a reconstruction sequence from the source sequence closely matches the effect of the $n$-product posterior channel on the non-product reconstruction sequence, manifesting as a block error constraint. This new formulation enables a single-letter asymptotic characterization for lossy source coding problems in terms of coherent information, quantum mutual information, and mutual information for fully quantum, QC, and classical setup, respectively \cite{atif2023lossy}. This has further inspired research on the rate-distortion-perception trade-off with a conditional distribution perception measure \cite{salehkalaibar2024rate} and quantum soft-covering lemma \cite{atif2023quantum}.

Motivated by these developments, we formulate the side information setting of the WZ incorporating a global error constraint using the notion of a posterior channel that produces the source from its reconstruction and side information available at the decoder. In this setting, instead of a single-letter distortion function, we are given a single-letter posterior channel capturing the nature of the loss. The goal is to construct an encoder and a decoder such that the joint effect of the producing reconstruction from the source sequence using side information available at the decoder is close to the effect of the $n$-product posterior channel acting on the non-product reconstruction sequence and side information, manifesting as a block error constraint. Similarly, we formulate the QC-QSI problem incorporating a global error constraint using the notion of a posterior classical-quantum (CQ) channel that produces the reference of the source and the side information using the reconstruction. We provide an inner bound in terms of the single-letter mutual information and quantum mutual information conditioned on side information for classical and QC cases, respectively (see Theorems \ref{thm:QC-QSI} and \ref{thm:C-CSI}). Furthermore, we provide the connection between RDT and rate-channel theory. In particular, we show that a lossy compression protocol using global error criterion also achieves optimal rate-distortion function for a specific distortion function and distortion level (see Theorem \ref{thm:connection}).

As for the inner bound, we use Winter's measurement compression protocol \cite{winter1999coding} to construct an encoding POVM, sequential decoder \cite{wilde2013sequential} using typical projectors for constructing decoding POVMs, and Sen's non-commutative union bound \cite{sen2012achieving}. For C-CSI, we employ likelihood encoders \cite{cuff2010coordination, atif2022source} for constructing randomized encoders. For detailed analysis, please refer to \cite{sohail2025WZ}.

\section{Main Results}
\label{sec:mainresults}
\textbf{Notations.} The set of density operators on Hilbert space $\calH_A$ is denoted by $\calD(\calH_A)$.
We denote the finite alphabet of a source as $\sfX$, and  
the set of probability distributions on the finite alphabet $\sfX$ as $\calP(\calX)$. Let $[\Theta] \deq \{1,2,\cdots,\Theta\}$. 

\begin{definition}\label{def:CQchannel}(Classical-Quantum (CQ) Channel)
Given a finite set $\sfX$ and a probability distribution $P_X$, a CQ channel $\mathcal{W}$ is   
specified by an ensemble of density of operators $\{(P_X(x),\calW_{x})\}_{ x \in \sfX}$. The corresponding average density operator is given as $\calW(P):= \sum_{x\in\sfX}P_X(x)\calW_x$.
\end{definition}

\subsection{Lossy Quantum-Classical Source Coding with Quantum Side Information}
Consider a quantum source $\rho^{AB}\in \calD(\calH_{A}\tensor \calH_{B})$ shared between a sender $A$ and a receiver $B$.
\begin{definition}(QC-QSI Source Coding Setup) A QC-QSI source coding setup is characterized by a triple $(\rho^{AB},\sfX,\calW_{\!\ \sfX\rightarrow \refe B})$, where $\rho^{AB}$ is the bipartite density operator of the source and its side information, $\sfX$ is the reconstruction alphabet, and $\calW:\sfX\rightarrow\calD(\calH_{\refe } \tensor \calH_B)$ is a single-letter posterior CQ channel.
\end{definition}
\begin{figure}[!htb]
    \centering
    \includegraphics[scale=0.85]{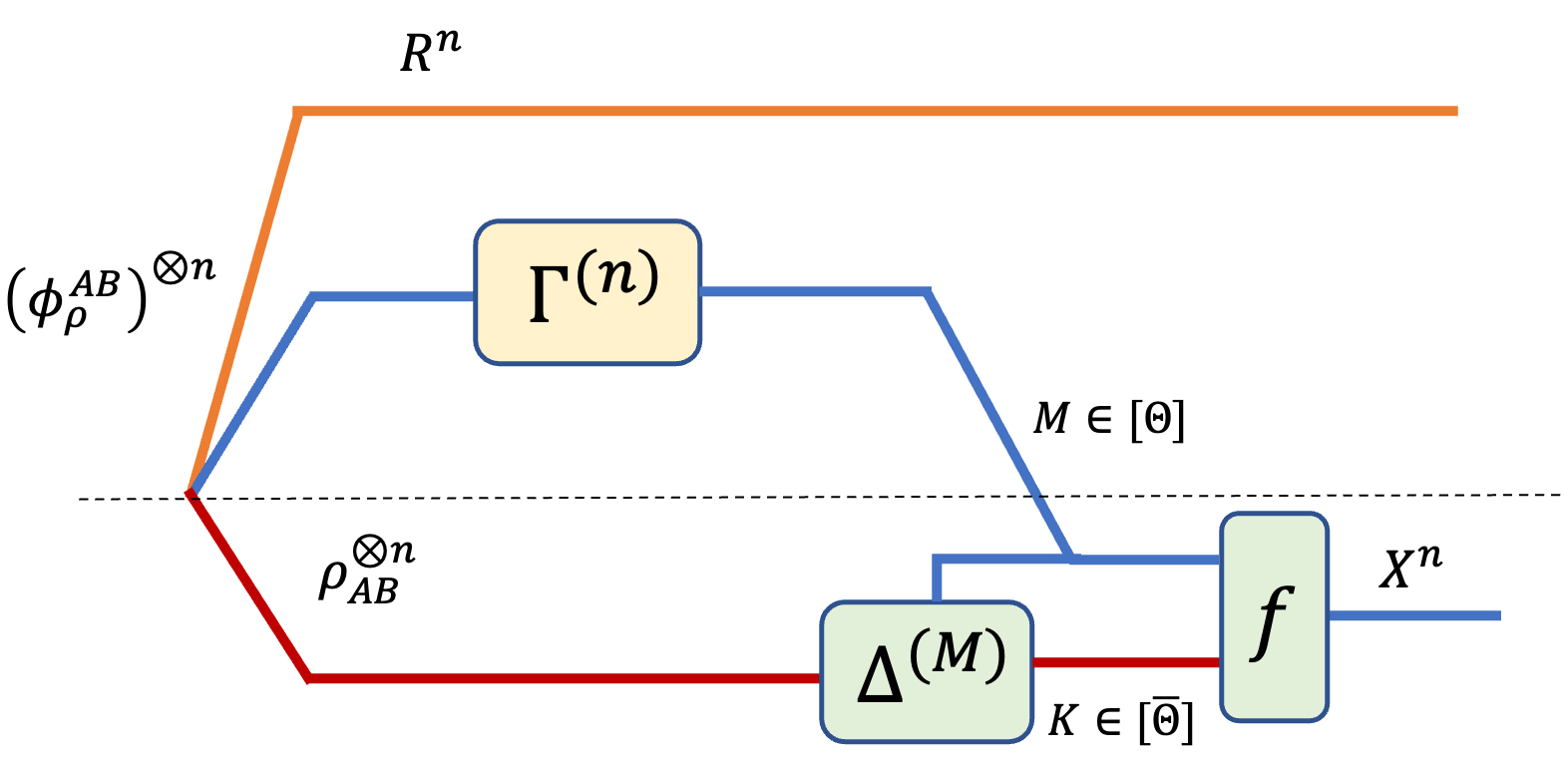}
    \vspace{-5pt}
    \caption{Illustration of Lossy QC-QSI Compression Protocol.}
    \label{fig:QC-QSI}
    \vspace{-5pt}
\end{figure}
\begin{definition}
    (Lossy QC-QSI Compression Protocol) For a given bipartite density operator $\rho^{AB}$ and a reconstruction alphabet $\sfX$, a $(n,\Theta)$ lossy QC-QSI compression protocol is characterized by $(i)$ an encoding POVM $\Gamma^{(n)}_{\sfA} \deq \{\epovm_{m}\}_{m=1}^{\Theta}$ acting on $A^n$, $(ii)$ For each $m\in[\Theta]$, a decoding POVM $\Gamma^{(n,m)}_{\sfB} \deq \{\dpovm^{(m)}_{k}\}_{k=1}^{\bTheta}$ acting on $B^n$, and  $(iii)$ a decoding map $f:[\Theta] \times [\bTheta] \rightarrow \sfX^n$, as shown in Figure \ref{fig:QC-QSI}. 
\end{definition}
\begin{definition}\label{def:qc_qsi_achievability}
    (Achievability) For a given QC-QSI source coding setup $(\rho^{AB},\sfX,\calW_{\sfX\rightarrow AB})$, a rate $R$ is said to be achievable if for all $\epsilon > 0$ and all sufficiently large $n$, there exists an $(n,\Theta)$ QC-QSI lossy compression protocol such that 
$\frac{1}{n}\log \Theta \leq R + \epsilon$, and $\Xi(\Gamma^{(n)}_\sfA,\Gamma^{(n)}_\sfB,f) \leq \epsilon$, where 
$\Xi(\Gamma^{(n)}_\sfA,\Gamma^{(n)}_\sfB,f) \deq$ 
\begin{align*}
    \Big\| \sum_{m,k} |\xn(m,k)\>\<\xn(m,k)| \tensor \tau^{\refe B}_{\xn(m,k)}   - \sum_{m,k}  Q_{\Xn}(\xn(m,k))|\xn(m,k)\>\<\xn(m,k)| \tensor \calW_{\xn(m,k)}^{RB}\Big\|_1,
\end{align*}
where $\calW_{\xn}^{RB} \deq \bigotimes_{i=1}^n \calW_{x_i}^{RB}$,  $\tau^{R^nB^n}_{\xn}\!\deq \!\Tr_{A^n}\{({I^{\refe}}^{\tensor n}\!\tensor \epovm_{m} \tensor \dpovm^{(m)}_k) (\phi^{RAB}_\rho)^{\tensor n}({I^{\refe A}}^{\tensor n}\! \tensor\!\dpovm^{(m)}_k)\}$ is the unnormalized system-induced density operator on systems $\refe^nB^n$, $, f(m,k) = \xn(m,k), \phi_{\rho}^{RAB}$ is the canonical purification of the state $\rho^{AB}$, and $Q_{\Xn}(\xn(m,k)) \deq \Tr\{(\Omega_{m} \tensor \Xi^{(m)}_k) (\rho^{AB})^{\tensor n}\}$ is the probability of observing the sequence $\xn(m,k)$.
\end{definition}
\begin{theorem}\label{thm:QC-QSI}For a given $(\rho_{AB},\sfX,\calW_{\sfX\rightarrow \refe B})$ QC-QSI source coding setup, 
a rate $R$ is achievable if $\calA(\rho^{AB},\calW_{X\rightarrow \refe B})$ is non-empty and
    $$R\geq I(X;\refe  B)_\sigma-I(X;B)_\sigma = I(X;\refe |B)_\sigma,$$
    where the quantum mutual information is computed with respect to the following quantum-classical state, $$\sigma^{XRB} \deq \sum_x P_X(x) |x\>\<x|^X \tensor \calW_x^{RB} \eqand  P_X(x) \in \calA(\rho^{AB},\calW_{X\rightarrow \refe B}),$$
    $\calA$ is the set of reconstruction distributions defined as 
$$\calA(\rho^{AB},\calW_{X\rightarrow \refe B}) \deq \{P_X\in\calP(\sfX):\sum_{x}P_X(x)\calW^{\refe B}_x = \Tr_{A}\{\phi^{\refe AB}\}\},$$ 
    $\phi^{\refe AB}  \text{ is a purification of $\rho^{AB}$, and } \{|x\>\}_{\{x\in \sfX\}}$ is an orthonormal basis 
    for the Hilbert space $\calH_X$ with $\dim{(\calH_X)}=|\calX|$.
\end{theorem}
\begin{proof}
    The proof is provided in Section \ref{sec:proof:QC-QSI}.
\end{proof}

\subsection{Lossy Classical Source Coding with Classical Side Information}
\begin{definition}(C-CSI Source Coding Setup) A C-CSI source coding setup is characterized by a triple $(\pxz,\sfY,W_{X|YZ})$, where $\pxz$ is the joint source and side-information distribution over a finite alphabet $\sfX\times\sfZ$, $\sfY$ is the reconstruction alphabet, and $W_{X|YZ}:\sfY\times\sfZ \rightarrow\sfX$ is the posterior (backward) channel, i.e., the single-letter conditional distribution of source given the reconstruction and side-information.
\end{definition}
\begin{definition}(Lossy C-CSI Source Compression Protocol) For a given $\pxz$ and reconstruction alphabet $\sfY$, an $(n,\Theta)$ lossy C-CSI source compression protocol consists of $(i)$ a randomized encoding map $\calE^{(n)}:\sfX^n\rightarrow[\Theta]$ and $(ii)$ a randomized decoding map $\calD^{(n)}:\sfZ^n \times[\Theta] \rightarrow\sfY^n,$ as shown in Figure \ref{fig:C-CSI}.  
\end{definition}
\vspace{-3pt}
\begin{figure}[!htb]
    \centering
    \includegraphics[scale=0.95]{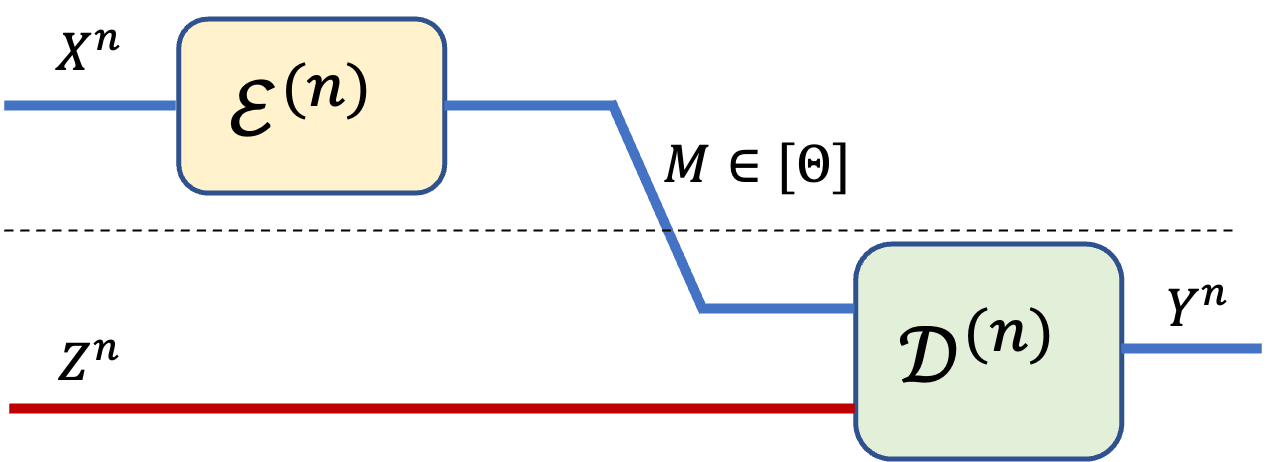}
    \vspace{-5pt}
    \caption{Illustration of Lossy C-CSI Compression Protocol.}
    \label{fig:C-CSI}
    \vspace{-3pt}
\end{figure}
\begin{definition}(Achievability) Given C-CSI source coding setup $(\pxz,\sfY,W_{X|YZ})$, a rate R is said to be achievable if for all $\epsilon >0$ and all sufficiently large $n$, there exists an $(n,\Theta)$ lossy compression protocol such that $\frac{1}{n}\log\Theta \leq R+\epsilon$ and $\Xi(\calE^{(n)},\calD^{(n)})\leq \epsilon$, where
\begin{align}
\Xi(\calE^{(n)},\calD^{(n)}) \deq \left\|P_{X^nY^nZ^n} \!-\! P_{Y^nZ^n}W_{X|YZ}^n\right\|_{\text{TV}},\nonumber
\vspace{-10pt}
\end{align}
    $P_{X^nY^nZ^n}(\xn\!,\!\yn\!,\!\zn) \! =\!  {\pxz^n(\xn\!,\!\zn)} \sum_{m \in [\Theta]}\encodern(m|x^n)$ $\decodern(y^n|m,\!\zn),\ \forall(\xn\!,\!\yn\!,\!\zn)\in \sfX^n \times \sfY^n \times \sfZ^n$, is the system-induced distribution,
 $P_{Y^nZ^n}W_{X|YZ}^n$ is the approximating distribution, and $W_{X|YZ}^n(\xn|\yn,\zn) \!\deq\! \prod_{i=1}^nW_{X|YZ}(x_i|y_i,z_i)$. Let $\calR_{\text{C-CSI}}(\pxz,\sfY,W_{X|YZ})$ denote the set of achievable rates.
\end{definition}
\begin{theorem}\label{thm:C-CSI}(Lossy C-CSI Compression Inner Bound)
    $R\in \calR_{\text{C-CSI}}$ if $\calA(P_{XZ},W_{X|YZ})$ is non-empty and there exists a PMF $P_{UXYZ} \in \calA$ such that 
    \begin{itemize}
        \item $P_{XYZ} = \sum_{u\in \sfU}P_{UXYZ}(u,\!x,\!y,\!z)$ for all $(x,\!y,\!z)$
        \item $Z-X-U, \ X-(U,Z)-Y,\text{ and } \ X-(Y,Z)-U$ are Markov chains
        \item $R\geq I(X;U)-I(U;Z),$
    \end{itemize}  
    where $\calA$ is the set of reconstruction distributions defined as 
    $$\calA(P_{XZ},W_{X|YZ})\deq \Big\{P_{U|X}, P_{Y|UZ}:P_{X|Z} \frac{\sum_{u\in\calU}P_{U|X}P_{Y|UZ}}{\sum_{u\in\sfU}P_{U|X}P_{Y|UZ}} = W_{X|YZ} \eqand X-(Y,Z)-U \Big\}.$$
\end{theorem}
\begin{proof}
    The proof is provided in Section \ref{sec:proof:C-CSI}.
\end{proof}

\subsection{Connection Between Rate-Channel Theory and Rate-Distortion Theory}
In \cite{sohail2023unique}, we have developed a new formulation of the lossy source coding problem called \textit{rate-channel theory} which is described below. 
\begin{definition}
    [Achievability]\label{def:clserror_constraint} Given a source coding setup $(\px,\sfY,W_{X|Y})$,
a rate $R$ is said to be achievable if for all $\epsilon >0$ and all sufficiently large $n$, there exists an 
$(n,\Theta)$ lossy source compression protocol consists of  $(i)$ a randomized encoding map $\encodern:\sfX^n \longrightarrow [\Theta]$ and 
$(ii)$ a randomized decoding map $\decodern:[\Theta] \longrightarrow {\sfY}^n$ such that $\frac{1}{n}\log \Theta \leq R + \epsilon$, and $\Xi(\encodern,\decodern) \leq \epsilon$, where 
\begin{equation}\label{def:error_constraint}
  \Xi(\encodern,\decodern)\deq \frac{1}{2}\sum_{\xn \yn}\|{P_{X^n\Yn}(x^n,\yn) -
    P_{\Yn}(\yn)
    \\
    \Pi_{i=1}^n\prevTC(x_i|\hat{x}_i)}\|,
\end{equation}
and
$
P_{X^n\Yn}(\xn,\yn) =  {\px^n(x^n)} \sum_{m \in [\Theta]} \encodern(m|x^n) \decodern(\yn|m), 
\text{  for all  } (x^n, \yn)\in \sfX^n \times {\sfY}^n,$ is the system-induced distribution,
 and $P_{\Yn}W_{X|Y}^n$ is the approximating distribution.
\end{definition}

\begin{theorem}\label{thm:Csourcecoding}(Rate-Channel Theory \cite[Theorem 2]{sohail2023unique}). For a $(\px,\sfY,W_{X|Y})$ source coding setup, a rate $R$ is said to be achievable if and only if $\calA(\px,W_{X|Y})$ is non-empty, and \begin{equation}\label{eqn:clsratedistortion}R \geq \min_{P_Y \in \calA(\px,W_{X|Y})} I(X;Y), \vspace{-5pt}\end{equation}
where $\calA(\px,W_{X|Y}) \deq \{P_Y \in \calP(\sfY): \!\! \text{for all }  x \in \sfX$, $ 
\sum_{y} P_{Y}(y) W_{X|Y}(x|y) = \px(x)\},$ is the set of reconstruction distributions.
\end{theorem}
Let us consider the case when the side information $Z$ is trivial. 
Given a lossy source coding setup $(\px,\sfY,W_{X|Y})$, 
consider a sequence of 
$(n,\Theta)$
lossy compression protocols 
that achieves the asymptotic performance  $R^\star \deq \min_{P_Y \in \calA(P_X,W_{X|Y})} I(X;Y)$
as given below in Theorem \ref{thm:Csourcecoding}.
Let 
$P_{X^nY^n}$ be the induced $n$-letter joint distribution on the source and the reconstruction vectors. 
Then, we see that 
$\lim_{n \rightarrow \infty} \frac{1}{n} \log \Theta =I(X;Y)$,
$$\lim_{n\rightarrow\infty}\|P_{X^nY^n} -
    P_{Y^n}W_{X|Y}^n\|_{\normalfont \text{TV}} = 0.$$


\begin{theorem}\label{thm:connection}
    Let $c>0$ and $b(x)$ be an arbitrary constant and a function, respectively. Consider a single-letter distortion function given as $d(x,y)= -c\log_2 W_{X|Y}(x|y)+b(x)$,
and distortion level $D={\mathbb{E}[d(X,Y)]}$ with respect to distribution $P_{Y}W_{X|Y}$, where $P_{Y} \in \calA(P_X,W_{X|Y})$ achieves the optimality in Theorem \ref{thm:Csourcecoding}. 
Then, the same sequence of protocols achieves the Shannon rate-distortion function at distortion level $D$ for the source $P_X$, and distortion function $d$, i.e., 
\vspace{-5pt}
\begin{align*}
    \lim_{n \rightarrow \infty} \!\mathbb{E} \bigg[\!\frac{1}{n} \!\sum_{i=1}^n d(X_i,Y_i)\!\bigg] \!=D,
\end{align*}
where the expectation is with respect to the distribution induced by the protocol $P_{X^nY^n}.$
\end{theorem}
\begin{proof}
    The proof is provided in Section \ref{sec:proof:connection}.
\end{proof}

\section{Proof of Theorem \ref{thm:QC-QSI}}
\label{sec:proof:QC-QSI}

For a given $(\rho^{AB},\sfX,\calW_{\!\ \sfX \rightarrow AB})$ QC-QSI source coding setup, we choose a $P_X(x) \in \calA(\rho^{AB},\calW_{\!\ \sfX \rightarrow AB})$. From now on, we let $\Theta := 2^{nR}$  and $\Theta := 2^{n\Rbar}$.

\noindent \textbf{Codebook Design}: We generate a codebook $\calC$ consisting of $n$-length codewords by randomly and independently selecting $\Theta\times\bTheta$ sequences $\calC\deq \{\Xn(m,k): m\in [\Theta]\eqand k\in [\bTheta]\}$ according to the following pruned distribution:
 \begin{align}\label{def:qc_distribution}
     &\codeDistribution(\Xn(m) = \xn) = \left\{\!\!\!\!\begin{array}{cc}
          \dfrac{P_X^n(\xn)}{(1-\varepsilon)}  & \mbox{for} \; \xn \in \Txqc\\
           0 &  \mbox{otherwise,}
     \end{array} \right. \!\!
 \end{align} 
  where $ P_X^n(\xn) = \prod_{i=1}^n P_X(x_i)$, $\Txqc$ is the $\delta$-typical set corresponding to the distribution $P_X$ on the set $\sfX$, and $\varepsilon(\delta,n) \triangleq \sum_{\xn \not \in \Txqc} P_X^n(\xn)$. Note that $\varepsilon(\delta,n) \searrow 0$ as $n \rightarrow \infty$ and for all sufficiently small $\delta > 0$. 
  
\vspace{2pt}
\noindent\textbf{Construction of Encoding POVM}:
Let $\Pi_{\rho}^{\refe B}$ and $\Pi_{\xn}^{\refe B}$ denote the $\delta$-typical and conditional $\delta$-typical projectors defined as in \cite[Def. 15.1.3]{wilde_arxivBook} and \cite[Def. 15.2.4]{wilde_arxivBook}, with respect to $\calW^{\refe B} \deq \sum_{x\in\sfX} P_X(x)\calW^{\refe B}_x$ and $\calW_{x}^{\refe B}$, respectively.
For all $\xn \in \Txqc$, define  
\begin{align*}
    \rhotilde_{\xn}^{\refe B} \deq \hat{\Pi} \Pi_{\rho}\Pi_{\xn}\calW_{\xn}^{\refe B}\Pi_{\xn}\Pi_{\rho}\hat{\Pi} \ \! \eqand \ \! \rhotilde^{\refe B} \deq \EE_{\PP}[\rhotilde_{\Xn}^{\refe B}],
\end{align*}
and $\rhotilde_{\xn}^{\refe B} = 0$ for $\xn \notin \Txqc$, where $\hat{\Pi}$ is the cut-off 
 projector onto the subspace spanned by the eigenbasis of $\EE[\Pi_{\rho}\Pi_{\Xn}\calW_{\Xn}^{\refe B}\Pi_{\Xn}\Pi_{\rho}]$ with eigenvalues greater than $\epsilon d$, where $d \!\deq\! 2^{-n(S(\calW^{\refe B})+\delta_1)}$ and $\delta_1$ will be specified later.  
 Using the Average Gentle Measurement Lemma \cite[Lemma 9.4.3]{wilde_arxivBook}, for any given $\epsilon \in (0,1)$, and all sufficiently large $n$ and all sufficiently small $\delta$, we have 
\begin{align} \label{eq:closeness_ref_SI}
    \EE_{\PP}[\|\rhotilde_{\Xn}^{\refe B}  - \calW_{\Xn}^{\refe B} \|_1] \leq \epsilon.
\end{align}
The proof follows from the derivation of \cite[Eq. 35]{wilde_e}. Using the above definitions, for all $\xn \in \sfX^n$, we construct the operators,
\vspace{-0.065in}
\[\epovm_{\xn}^{\refe B} \deq \gamma_{\xn}\ {(\calW^{{\refe B}^{\tensor n}})}^{-1/2} \rhotilde_{\xn}^{\refe B} {(\calW^{{\refe B}^{\tensor n}})}^{-1/2} ,\text{ where }\gamma_{\xn} \!\deq\! \gamma \cdot |\{(m,k)\!:\!\Xn(m,k)\! =\! \xn\}|,\]
 $\gamma \deq (\Theta\bTheta)^{-1}  \frac{(1-\varepsilon)}{(1+\eta)}$ and $\eta \in (0,1)$ is a parameter that determines the probability of not obtaining a sub-POVM. Note that in the above definition operator  $\epovm_{\xn}^{\refe B}$ acts on $(\calH_{\refe^n} \tensor \calH_{\Bn})$, however, we define $\epovm^{A}_{\xn} \in \calL(\calH_{\An})$. To obtain this, we transform $\epovm_{\xn}^{\refe B}$ as 
\[\epovm_{\xn}^{A} = \sum_{\an \abarn} \<\an|\epovm_{\xn}^{\refe B} \ |\abarn\>_{\refe B} |\an\>\<\abarn|_{A},\]
where $\{|a\>_A\}$ is an eigenbasis of $\rho^A$, $\{|a\>_{RB}\}$ is an eigenbasis of $\rho^{\refe B} := \Tr_A\{|\phi_{\refe AB}\>\<\phi_{\refe AB}|\}$, and $|\phi_{\refe AB}\>$ is the canonical purification of $\rho^{AB}$. Furthermore, by using the equivalence of purification \cite[Thm. 5.1.1]{wilde_arxivBook}, it can be easily shown that $\Tr\{\epovm_{\xn}^A\rho^{A^{\tensor n}}\} = \Tr\{\epovm_{\xn}^{\refe B} \ \calW^{\refe B ^{\tensor n}}\}.$

Let $\I_{\{\mbox{sP}\}}$ denote the indicator random variable corresponding to the event that  $\{\epovm_{\xn}^{A} \colon \xn \in  \Txqc\}$ forms a  sub-POVM. We now provide a proposition from \cite{winter}, which will be helpful later in the analysis.
\begin{prop} \label{prop:enc_subpovm}For all $\epsilon, \eta \in (0,1)$, for all sufficiently small $\delta > 0$, and for all sufficiently large $n$, we have
$\EE[{\I_{\{\mbox{\normalfont sP}\}}}] \geq 1-\epsilon$, if $\frac{1}{n}(\log(\Theta)+\log(\bTheta)) > \chi(\{P_X(x),\calW^{\refe B}_x\})$.
\end{prop}
If $\I_{\{\mbox{sP}\}} = 1$, then construct sub-POVM $\Gamma^{(n)}_{\sfA}$ as follows: $\Gamma^{(n)}_{\sfA} \deq \big\{\sum_{k\in[\bTheta]}\epovm_{\xn(m,k)}^{A}\big\}_{m\in[\Theta]}.$ 
We then add an additional operator $\epovm_{0}^A \deq (I\!-\!\sum_{m\in[\Theta]}\sum_{k\in [\bTheta]}\epovm_{\xn(m,k)}^A)$, associated with an arbitrary sequence $\xn_0 \in \sfX^n \backslash\Txqc$, to form a valid POVM $[\Gamma^{(n)}_\sfA]$ with at most $(\Theta \times \bTheta+1)$ elements. The extra element $\epovm_{\xn_0}^A$ corresponds to a failed encoding.

\vspace{2pt}
\noindent\textbf{Construction of Decoding POVM}:
For the ensemble $\{P_X(x),\calW^B_x\}$, we construct a collection of $n$-letter Bob's POVMs, one for each $m \in [\Theta]$, capable of decoding the message $k\in [\bTheta]$. 
Upon receiving the message $m$, Bob performs a sequence of binary measurements $\{\Pi_{\xn(m,k)},(I-\Pi_{\xn(m,k)})\}$ for all sequence $\xn(m,k) \in \calC$, where $\Pi_{\xn(m,k)}$ is a conditional typical projector for the tensor-product state $\calW^B_{\xn(m,k)}$.
Define the decoding POVM element as
$$\dpovm_{k}^{(m)} \deq \Pibar^{(m)}_{1} \cdots \Pibar^{(m)}_{k-1} \
\Pi^{(m)}_{k} \ \Pibar^{(m)}_{k-1} \cdots \Pibar^{(m)}_{1},$$
where $ \Pibar^{(m)}_{k}$ and $ \Pi^{(m)}_{k}$ are the shorthand notation for $(I-\Pi_{\xn(m,k)})$ and $\Pi_{\xn(m,k)}$, respectively.
The following proposition demonstrates that for Bob's POVMs, we can make the average probability of error arbitrarily small by using the non-commutative union bound \cite{sen2012achieving}.
\begin{prop}\label{prop:qc_packing}
    Given the ensemble $\{P_X(x),\calW^B_x\}$ and the collection of POVMs $\{\Xi_{k}^{(m)}\}_{k\in[\bTheta]}$, for any $\epsilon \in (0,1)$
$$\EE_\PP\left[\frac{1}{\bTheta}\sum_{k\in[\bTheta]} \Tr\left\{(I-\dpovm_{k}^{(m)})\calW_{m,k}^B\right\}\right] \leq \epsilon,$$
for sufficiently small $\delta>0$ and for all sufficiently large $n$, and for all $m \in [\Theta]$, if $\frac{1}{n}\log(\bTheta) < \chi(\{P_X(x),\calW^B_x\})$. 
\end{prop}
\begin{proof}
    The proof follows from packing lemma using sequential decoding \cite[Sec. 16.6]{wilde_arxivBook}, while making the following identification. For each $m\in\calM$, identify $\calM$ as $[\bTheta]$, $\calX$ as $\Txqc$, $\{\sigma_{C_m}\}_m$ with $\{\calW_{m,k}^{\refe B}\}_k$, $\Pi_x$ with $\Pi_{k}^{(m)}$, $d$ with $2^{n(S(X|B)_{\tau}-\Bar{\delta})}$, and $D$ with $2^{n(S(B)_\tau-\Bar{\delta})}$, where $\tau^{XB} := \sum_x P_X(x) |x\>\<x|\tensor \calW^B_x$ and $\Bar{\delta} \searrow 0$ as $\delta\searrow 0$. 
\end{proof}
In general, the decoding POVM elements satisfy the condition $\sum_{k \in [\bTheta]} \dpovm_{k}^{(m)} \leq I$ for all $m \in [\Theta].$
Under the condition $\{\I_{\{\mbox{sP}\}} = 1\}$,
construct sub-POVM $\Gamma_\sfB^{(n)}$ as follows: $\{\dpovm_{k}^{(m)}\}_{k\in [\bTheta]}$ for each $m\in[\Theta]$. This sub-POVM is completed by adding an additional operator $\dpovm_{0}^{(m)}\deq(I\!-\!\sum_{k\in[\bTheta]}\dpovm_k^{(m)})$ to form a valid POVM $[\Gamma_\sfB^{(n}]$, for each $m\in [\Theta]$.

\noindent\textbf{Error Analysis}: We begin by defining the following code-dependent random variables \(E_1\), \(E_2\), and \(E_3\), which will be useful in the error analysis, given as:
\begin{align*}
    E_1 \deq \sum_{m\in [\Theta]}\sum_{k\in [\bTheta]} {(\Theta\bTheta)}^{-1}\ &\Tr\{\rhotilde_{m,k}^{\refe B}\}\eqand E_2 \deq \sum_{m\in [\Theta]}\sum_{k\in [\bTheta]}{(\Theta\bTheta)}^{-1}\ \|\rhotilde_{m,k}^{\refe B} - \calW^{\refe B}_{m,k}\|_1,
\end{align*}
where $\rhotilde_{m,k}^{\refe B}$ and $\calW_{m,k}^{\refe B}$ are the shorthand notation for $\rhotilde_{\xn(m,k)}^{\refe B}$ and $\calW_{\xn(m,k)}^{\refe B},$ respectively. We provide the following proposition that bound these terms under the condition $\I_{\curly{\mbox{\normalfont sP}}} = 1$.
\begin{prop}\label{prop:code_dependent_RV}
For all $\epsilon\in(0,1)$, for all sufficiently small $\eta, \delta>0$, and for all sufficiently large $n$, we have $\EE_\PP[E_1]\geq (1-\epsilon) \eqand  \EE_\PP[E_2]\leq \epsilon$.
\end{prop}
\begin{proof}
    The proof is provided in Appendix \ref{app:prop:code_dependent_RV}.
\end{proof}
Now, Observe that after applying the binary projectors Bob generates $(m,k')$, and consequently the sequence $\xn(m,k')$ using the decoding map $f$. The (unnormalized) post-measured state after the sequential decoding can be expressed as
$$({I^{\refe}}^{\tensor n}\tensor \Pi_{k}^{(m)}\Pibar_{k-1}^{(m)}\cdots \Pibar_{1}^{(m)})\omega^{\refe B}_{m,k} ({I^{\refe}}^{\tensor n}\tensor \Pibar_{1}^{(m)} \cdots \Pibar_{k-1}^{(m)}\Pi_{k}^{(m)}),$$ where 
\vspace{-10pt}\begin{equation}\label{eqn:omegamk}
    \omega_{m,k}^{\refe B} \deq \Tr_{A^n}\{({I^{\refe B}}^{\tensor n}\!\tensor \epovm_{m,k}^{A})({\phi^{\refe AB}}^{\tensor n})\} 
\end{equation}
is the unnormalized post-measured state from Alice's encoding. Here, $\omega_{m,k}^{\refe B} \eqand  \epovm_{m,k}^{A}$ are used as a shorthand notation for $\omega_{\xn(m,k)}^{\refe B} \eqand \epovm_{\xn(m,k)}^{A}$, respectively, and  $\phi^{\refe A B}:= |\phi\>\<\phi|^{\refe AB}$ and $|\phi\>^{\refe AB} $ is the canonical purification of $\rho^{AB}$. Furthermore, note that the projectors $\Pi_{k}^{(m)}\cdots \Pi_{1}^{(m)}$ and the POVM element ${\Lambda_{k}^{(m)}}$ are related by a polar decomposition, given as
\begin{align*}
    \sqrt{\Lambda_{k}^{(m)}} = U_{m,k} \ \Pi_{k}^{(m)}\Pibar_{k-1}^{(m)}\cdots \Pibar_{1}^{(m)},
\end{align*}
for some unitary $U_{m,k}$. Therefore, Bob first applies the binary projectors and constructs $\xn(m,k')$. Following this, Bob applies the unitary $U_{m,k}$ and the (unnormalized) state becomes as follows:
\begin{equation}\label{eqn:lambdamk}
    \lambda_{(m,k),k'}^{\refe B}\deq\Big({I^{\refe}}^{\tensor n}\tensor \sqrt{\Lambda_{k'}^{(m)}}\Big)\omega^{\refe B}_{m,k}\Big({I^{\refe}}^{\tensor n}\tensor \sqrt{\Lambda_{k'}^{(m)}}\Big).
\end{equation}
If $k=k'$ (indicating correct decoding), then $\lambda_{(m,k),k'}^{\refe B} = \lambda_{m,k}^{\refe B}$, i.e., Bob successfully recovers the sequence $\xn(m,k)$.
Now, following Definition \ref{def:qc_qsi_achievability}, our objective is to show that the following term 
\begin{align*}
    \EE_\PP[&\error] \\&= \EE_\PP\Bigg[ \Big\|\sum_{m \in [\Theta]\cup \{0\}}\sum_{\substack{k\in [\bTheta]\cup \{0\}}}\sum_{k'\in [\bTheta]\cup \{0\}} \!\!|\xn(m,k')\>\<\xn(m,k')|\tensor \big(\lambda_{(m,k),k'}^{\refe B} -  \Tr\{\lambda_{(m,k),k'}^{\refe B}\}\calW_{m,k'}^{\refe B}\big)\Big\|_1\Bigg]
\end{align*}
can be made arbitrarily small for sufficiently large $n$ for the code $\codebook$, where $\calW_{m,k'}^{\refe B}\deq \calW_{\xn(m,k')}^{\refe B}.$ 

First, we split the error $ \error $ into two terms using the indicator function $\I_{\curly{\mbox{\normalfont sP}}}$ as 
\begin{align}
\error&=\I_{\curly{\mbox{\normalfont sP}}} 
\error + \round{1- \I_{\curly{\mbox{\normalfont sP}}}}\error\\
&\leq \I_{\curly{\mbox{\normalfont sP}}} 
\error  + 2\round{1- \I_{\curly{\mbox{\normalfont sP}}}} \label{eqn:qcerrorsubpovm},
\end{align}
where \eqref{eqn:qcerrorsubpovm} follows from upper bounding the trace distance between two density operators by its maximum value of two. 
Under the condition $\I_{\curly{\mbox{\normalfont sP}}} = 1$, 
 \begin{align*}
     \error &\overset{}{\leq} \underbrace{\sum_{m \in [\Theta]}\sum_{\substack{k\in [\bTheta]}}\|\lambda_{m,k}^{\refe B} -  \Tr\{\lambda_{m,k}^{\refe B}\}\calW_{m,k}^{\refe B}\|_1}_{\zeta_{\text{CP}}} 
     + \ 2 \!\!\!\!\underbrace{\sum_{\substack{k'\in [\bTheta]\cup \{0\}}}\!\!\!\! \Tr\{\lambda_{(0,0),k'}^{\refe B}\}}_{\zeta_{\text{NC}}}
     \\
     &\hspace{50pt}+2\!\!\underbrace{\sum_{m \in [\Theta]}\sum_{\substack{k \neq k' \in [\bTheta]}}  \Tr\{\lambda_{(m,k),k'}^{\refe B}\}}_{\zeta_{\text{NP}_1}}  +   \ 2\underbrace{\sum_{m \in [\Theta]}\sum_{\substack{k\in [\bTheta]}} \Tr\{\lambda_{(m,k),0}^{\refe B}\}}_{\zeta_{\text{NP}_2}},
 \end{align*}
 where the inequality follows from the fact that $\|\calW_{\xn}\|_1=1$ for all $\xn\in \sfX^n$ and  using triangle inequality.

\noindent \textbf{Step 1. Bounding the error induced by not covering, i.e., encoding error.} 

\noindent The error term $\zeta_{\text{NC}}$ captures the error induced by not covering the $n$-tensored posterior reference channel. We provide the following proposition that bounds this term.  
\begin{prop}\label{prop:qc_NC}
    For all $\epsilon\in(0,1)$, for all sufficiently small $\eta, \delta>0$, and for all sufficiently large $n$, we have $\EE_\PP[\I_{\curly{\mbox{\normalfont sP}}}\zeta_{\text{NC}}]\leq \epsilon$.
\end{prop}
\begin{proof}
The proof is provided in Appendix \ref{app:prop:proof:qc_NC}.
\end{proof}
\noindent \textbf{Step 2. Bounding the error induced by not packing, i.e., decoding error.} 

\noindent The error term $\zeta_{\text{NP}} \deq \zeta_{\text{NP}_1} + \zeta_{\text{NP}_2}$ captures the error induced by not packing, i.e., incorrect decoding. Therefore, $\zeta_{\zeta_{\text{NP}}}$ can rewritten as the $(1-$ probability of correct decoding), given as
\begin{align}
\zeta_{\text{NP}} &= \sum_{m \in [\Theta]}\sum_{\substack{k \neq k' \in [\bTheta]}}  \Tr\{\lambda_{(m,k),k'}^{\refe B}\} +  \sum_{m \in [\Theta]}\sum_{\substack{k\in [\bTheta]}} \Tr\{\lambda_{(m,k),0}^{\refe B}\}\nonumber\\
&= \sum_{m \in [\Theta]}\sum_{\substack{k\in [\bTheta]}} \Big(\Tr\big\{\big({I^{\refe}}^{\tensor n} \!\tensor\!\!\!\!\sum_{\substack{k' \neq k \in [\bTheta]}} \!\!\!\!\dpovm_{k'}^{(m)} \big)\omega^{RB}_{m,k}\big\}   +   \Tr\big\{\big({I^{\refe}}^{\tensor n} \!\!\tensor \dpovm_{0}^{(m)} \big)\omega^{RB}_{m,k}\big\}\}\Big)\nonumber\\
& = \sum_{m \in [\Theta]}\sum_{\substack{k\in [\bTheta]}}
\Tr\big\{\big(I-\dpovm_{k}^{(m)} \big)\omega^{B}_{m,k}\big\}\}\label{eqn:zetaNP},
\end{align}
where the second equality follows from \eqref{eqn:lambdamk} and $\omega^B_{m,k} = \Tr_{\refe^n}\{\omega^{\refe B}_{m,k}\}$. We provide the following proposition that bounds $\zeta_{\text{NP}}$.
\begin{prop}\label{prop:qc_NP}
    For all $\epsilon\in(0,1)$, for all sufficiently small $\eta, \delta>0$, and for all sufficiently large $n$, we have $\EE_\PP[\I_{\curly{\mbox{\normalfont sP}}}\zeta_{\text{NP}}]\leq 2\epsilon$.
\end{prop}
\begin{proof}
The proof is provided in Appendix \ref{app:prop:proof:qc_NP}.
\end{proof}
\noindent \textbf{Step 2. Bounding the error induced by covering and packing.} 

\noindent The error term $\zeta_{\text{CP}}$ captures the error induced by covering and packing. Consider the following inequalities:
\begin{align*}
    \zeta_{\text{CP}} 
    &\overset{a}{\leq} \sum_{m \in [\Theta]}\sum_{\substack{k\in [\bTheta]}} \Tr\{\omega^{\refe B}_{m,k}\}\|(\Tr\{\omega^{\refe B}_{m,k}\})^{-1}\lambda_{m,k}^{\refe B} -  \omegabar^{\refe B}_{m,k}\|_1 \\
    &\hspace{75pt}+ \sum_{m \in [\Theta]}\sum_{\substack{k\in [\bTheta]}} \Tr\{\omega^{\refe B}_{m,k}\} \|\omegabar^{\refe B}_{m,k} - (\Tr\{\omega^{\refe B}_{m,k}\})^{-1}\Tr\{\lambda_{m,k}^{\refe B}\}\omegabar^{\refe B}_{m,k}\|_1
    \\
    &\hspace{150pt}+ \sum_{m \in [\Theta]}\sum_{\substack{k\in [\bTheta]}}\Tr\{\lambda_{m,k}^{\refe B}\}\|\omegabar_{m,k}^{\refe B} -  \calW_{m,k}^{\refe B}\|_1\\
    &\overset{b}{=} \sum_{m \in [\Theta]}\sum_{\substack{k\in [\bTheta]}}\Tr\{\omega^{\refe B}_{m,k}\} \Big\|\Big({I^{\refe}}^{\tensor n} \!\tensor \sqrt{\Lambda_{k}^{(m)}}\Big)\omegabar^{\refe B}_{m,k}\Big({I^{\refe}}^{\tensor n} \!\tensor \sqrt{\Lambda_{k}^{(m)}}\Big) -  \omegabar^{\refe B}_{m,k}\Big\|_1 
    \\
    &\hspace{75pt} + \sum_{m \in [\Theta]}\sum_{\substack{k\in [\bTheta]}}\big(\Tr\{\omega^{\refe B}_{m,k}\}-\Tr\{({I^{\refe}}^{\tensor n} \!\tensor \Lambda_{k}^{(m)})\omega_{m,k}^{\refe B}\}\big) \|\omegabar^{\refe B}_{m,k}\|_1 + \zeta_{\text{C}}\\
    &\overset{c}{\leq}
    \sum_{m \in [\Theta]}\sum_{\substack{k\in [\bTheta]}} 2 \ \Tr\{\omega^{\refe B}_{m,k}\} \sqrt{\Tr\Big\{\Big( I-\big({I^{\refe}}^{\tensor n} \!\tensor \Lambda_{k}^{(m)}\big)\Big)\omegabar^{\refe B}_{m,k}\Big\}} 
    \\&\hspace{75pt}+\sum_{m \in [\Theta]}\sum_{\substack{k\in [\bTheta]}}\Tr\Big\{\Big( I-\big({I^{\refe}}^{\tensor n} \!\tensor \Lambda_{k}^{(m)}\big)\Big) \omega^{\refe B}_{m,k}\Big\}
+ \zeta_{\text{C}}\\
    &\overset{d}{\leq}
    2 \sqrt{\sum_{m \in [\Theta]}\sum_{\substack{k\in [\bTheta]}}  \Tr\{\omega^{\refe B}_{m,k}\} \Tr\Big\{\Big( I-\big({I^{\refe}}^{\tensor n} \!\tensor \Lambda_{k}^{(m)}\big)\Big)\omegabar^{\refe B}_{m,k}\Big\}} + \zeta_{\text{NP}} + \zeta_{\text{C}}\\
    &\overset{e}{\leq} 3\sqrt{\zeta_{\text{NP}}}+ \zeta_{\text{C}},
\end{align*}
where $(a)$ follows from adding and subtracting appropriate terms, defining $\omegabar^{\refe B}_{m,k}\deq (\Tr\{\omega^{\refe B}_{m,k}\})^{-1} \omega^{\refe B}_{m,k}$, and applying triangle inequality, $(b)$ follows from \eqref{eqn:lambdamk} and the definition $$\zeta_{\text{C}}\deq \sum_{m \in [\Theta]}\sum_{\substack{k\in [\bTheta]}}\Tr\{\lambda_{m,k}^{\refe B}\}\|\omegabar_{m,k}^{\refe B} -  \calW_{m,k}^{\refe B}\|_1,$$ $(c)$ follows from Gentle Operator Lemma \cite[Lemma 9.4.2]{wilde_arxivBook}, $(d)$ follows from \eqref{eqn:zetaNP} and applying Jensen's inequality for concave functions, and $(e)$ is based on the fact that $x \leq \sqrt{x}$ for all $x\in [0,1]$.
We provide the following proposition that bounds $\zeta_{\text{C}}$.
\begin{prop}\label{prop:qc_cov}
    For all $\epsilon\in(0,1)$, for all sufficiently small $\eta, \delta>0$, and for all sufficiently large $n$, we have $\EE_\PP[\I_{\curly{\mbox{\normalfont sP}}}\zeta_{\text{C}}]\leq 2\epsilon$.
\end{prop} 
\begin{proof}
The proof is provided in Appendix \ref{app:prop:proof:qc_cov}.
\end{proof}
Using Propositions \ref{prop:qc_NP} and \ref{prop:qc_cov}, we now bound the error term $\zeta_{\text{CP}}$. For all $\epsilon\in (0,1)$
\begin{align}
\EE_\PP[&\I_{\curly{\mbox{\normalfont sP}}}\zeta_{\text{CP}}] = \EE_\PP[3\sqrt{\zeta_{\text{NP}}}+\zeta_{\text{C}}] \leq 3\sqrt{\EE_\PP[\I_{\curly{\mbox{\normalfont sP}}} \zeta_{\text{NP}}}]+\EE_{\PP}[\I_{\curly{\mbox{\normalfont sP}}}\zeta_{\text{C}}] \leq 3\sqrt{2\epsilon}+2\epsilon\label{eqn:zetaCP}
\end{align}
where the first inequality follows from Jensen's inequality for concave functions. Finally, using Propositions \ref{prop:qc_NC} and  \ref{prop:qc_NP}, and \eqref{eqn:zetaCP}, we bound $\error$, for all $\epsilon\in(0,1),$
\begin{align*}
    \EE_\PP[\error]&\leq \EE_\PP[\I_{\curly{\mbox{\normalfont sP}}} 
\error  + 2(1- \I_{\curly{\mbox{\normalfont sP}}})] \\&\leq \EE_\PP[\I_{\curly{\mbox{\normalfont sP}}} 
\error] +2\epsilon \leq 3\sqrt{2\epsilon}+10\epsilon.
\end{align*}
Since $\EE_\PP[\error]\leq 3\sqrt{2\epsilon}+10\epsilon$, there exists a codebook $\codebook$ and the associated POVMs $\Gamma^{(n)}_\sfA$ and $\Gamma^{(n)}_\sfB$ such that $\error\leq 3\sqrt{2\epsilon}+10\epsilon$. This completes the proof of Theorem \ref{thm:QC-QSI}.

\section{Proof of Theorem \ref{thm:C-CSI}}
\label{sec:proof:C-CSI}
\noindent For a given $(\pxz,\sfY,W_{X|YZ})$ C-CSI source coding setup, we choose the distributions $(P_{U|X},P_{Y|UZ}) \in \calA(P_{XZ},W_{X|YZ})$. 

\noindent \textbf{Codebook Construction}:
We generate a codebook $\codebook$ consisting of $n$-length codewords by randomly and independently selecting $2^{n\Rbar}$ sequences $\{\Un(1),\Un(2),\cdots, \Un(2^{n\Rbar})\}$ according to the following pruned distribution:
 \begin{align}\label{def:c_distribution}
     &\codeDistribution(\Un(m) = \un) = \left\{\!\!\!\!\begin{array}{cc}
          \dfrac{\targetpu^n(\un)}{(1-\varepsilon)}  & \mbox{for} \; \un \in \Tuqc\\
           0 &  \mbox{otherwise,}
     \end{array} \right. \!\!
 \end{align} 
  where $ \targetpu^n(\un) = \prod_{i=1}^n \targetpu(u_i)$, $\Tuqc$ is the $\delta$-typical set corresponding to the distribution $\pu$ on the set $\calU$, and $\varepsilon(\delta,n) \triangleq \sum_{\un \not \in \Tuqc} \targetpu^n(\un)$. Note that $\varepsilon(\delta,n) \searrow 0$ as $n \rightarrow \infty$ and for all sufficiently small $\delta > 0$. 
 The generated codebook $\calC$ is revealed to the encoder and decoder before the C-CSI protocol begins.

\vspace{5pt}
\noindent \textbf{Encoder Description}:
For an observed source sequence $\xn$, construct a randomized encoder that chooses an index $l \in [2^{n\Rbar}]$ according to a sub-PMF $E_{L|X^n}(l|x^n)$
, which is analogous to the likelihood encoders used in  \cite{cuff2013distributed, atif2022source,sohail2023unique}. We now specify $E_{L|X^n}(l|x^n)$ for $x^n\in \Tx$ and $l\in[2^{n\Rbar}]$, where $\hat{\delta} = \delta(|\calX| + |\calU|)$. For a  $\eta \in (0,1)$ (to be specified later), and $\delta>0$, define
\begin{align}
 E_{L|X^n}&(l|x^n) \deq \sum_{\un}   \frac{1}{2^{n\Rbar}}\frac{(1-\varepsilon)}{(1+\eta)}\frac{P^n_{X|U}(x^n|u^n)}{\px^n(x^n)}\I_{\{\xn\in \Tx\}} 
  \I_{\{\un \in \Tucond\}}
  \I_{\{\Un(l) = \un\}},\nonumber
\end{align}
Let $\Ipmf$ denotes the indicator random variable corresponding to the event 
that $\{E_{L|\Xn}(l|\xn)\}_{l \in [\Theta]}$ forms a sub-PMF for all $\xn\in\Tx $. Once the index $l$ is chosen, it gets mapped to an index $m \in [2^{nR}]$. The mapping is done using a binning map $\calB:[2^{n\Rbar}] \rightarrow [2^{nR}]$. To summarize, on observing $\xn$, the encoder chooses $L \in [2^{n\Rbar}]$ stochastically according the PMF $E_{L|\Xn}(\cdot|\xn)$, and communicate the index $\calB(L)$ to the decoder. 

\vspace{3pt}
\noindent After specifying the PMF $E_{L|\Xn}(\cdot|\xn)$, we now characterize $P_{M|\Xn}$. If $\Ipmf \!\! = \!\! 1$, then construct the sub-PMF $P_{M|\Xn}(m|\xn) \!\deq\! \sum_{l \in [2^{n\Rbar}]}E_{L|\Xn}(l|\xn) \I_{\set{\calB(l) = m}}$, for all $\xn\in\Tx \eqand l\in [2^{n\Rbar}].$ We then append an additional PMF element $P_{L|\Xn}(0|\xn) = E_{L|\Xn}(0|\xn) \deq 1-\sum_{l\in[2^{n\Rbar}]} E_{L|X^n}(l|x^n)$ for all $\xn \in \Tx$, associated with $m=0$, to form a valid PMF $P_{M|\Xn}(m|\xn)$ for all $\xn \in \Tx$ and $ m\in\set{0} \cup [2^{nR}]$. If $\xn \not \in \Tx$, then we define $P_{M|\Xn}(m|\xn) = \I_{\set{m=0}}$. Finally, if $\Ipmf = 0, \text{ then } P_{M|\Xn}(m|\xn) = \I_{\{m=0\}}$, for all $\xn \in \calX^n$. This concludes the encoder description. We provide a proposition from \cite{atif2023lossy}, which will be helpful later in the analysis.

\vspace{3pt}
\begin{proposition}\label{prop:clssubPMF}
    For all $\epsilon,\eta \in (0,1)$, for all sufficiently small $\delta > 0$, and sufficiently large $n$, we have $\EE[{\Ipmf}] \geq 1-\epsilon$,
if $\Rbar > I(X;U)$.
\end{proposition}

\vspace{3pt}
\noindent \textbf{Decoder Description}: 
For an observed sequence $m \in \set{0} \cup [2^{nR}]$ communicated by the encoder and the sequence $\zn \in \calZ^n$, the decoder constructs the following set:
$\sfL(m,\zn) \deq \set{l \in [2^{n\Rbar}] : \calB(l) = m \eqand (\Un(l),\zn) \in \Tuz}.$
After this, the decoder outputs $\calD(m,\zn) = \Un(l)$ if $\sfL(m,\zn) = \set{l} \eqand m\neq 0$. Otherwise, the decoder outputs a fixed $\un_0 \in \calU^n \backslash \Tu$. At the end, the decoder chooses $\yn$ according to PMF $P^n_{Y|UZ}(\yn|\calD(m,\zn),\zn)$. This implies the PMF $P_{Y^n|MZ^n}(\cdot|m,\zn)$ can be expressed as: 
\[P_{Y^n|MZ^n}(\cdot|m,\zn) = P^n_{Y|UZ}(\cdot|\calD(m,\zn),\zn).\]

\vspace{3pt}
\noindent \textbf{Error Analysis}: 
We show that for the above-mentioned encoder and decoder, $P_{\Xn\Zn\Yn}$ is close to the approximating distribution $P_{\Yn\Zn} W^n_{X|YZ}$. These PMFs can be further expressed as follows:
\begin{align*}
    P_{\Xn\Yn\Zn}(\xn,\yn,\zn) &= \sum_{m \in [\theta] \cup \set{0}} P_{XZ}^n(\xn,\zn) P_{M|\Xn}(m|\xn)P_{\Yn|M\Zn}(\yn|m,\zn)\\
    P_{\Yn\Zn}(\yn,\zn) W^n_{X|YZ}(\xn
|\yn,\zn) &= \sum_{\sfxn}\sum_{\substack{m \in [\theta] \cup \set{0}}}  P^n_{XZ}(\sfxn,\zn)P_{M|\Xn}(m|\sfxn) P_{\Yn|M\Zn}(\yn|m,\zn)\\
& \hspace{2.5in}\times W^n_{X|YZ}(\xn
|\yn,\zn),
\end{align*}
where $P_{M|\Xn}(m|\sfxn) \eqand P_{\Yn|M\Zn}(\yn|m,\zn)$ are PMFs induced by encoder and decoder, respectively, and for convenience, we denote $\theta \deq [2^{nR}]$ and  $\bar{\theta} \deq [2^{n\Rbar}]$, for the remaining of the paper. We begin by splitting the error $ \Xi(\encodern,\decodern)$ into two terms using $\Ipmf$ as 
\begin{align}
    \Xi(\encodern,\decodern) 
    &= \Ipmf 
\Xi(\encodern,\decodern) + (1- \Ipmf)\Xi(\encodern,\decodern),\nonumber \\
&\leq \Ipmf 
\Xi(\encodern,\decodern) + ({1- \Ipmf}),\label{eqn:clserrorsubpmf}
\end{align}
where \eqref{eqn:clserrorsubpmf} follows from upper bounding the total variation between two PMFs by one, i.e., the maximum value of the 
total variation between two PMFs. 
Using the triangle inequality, we now expand $\Xi(\encodern,\decodern)$. Under the condition $\Ipmf = 1$, as follows:
\begin{align*}
    2 \ \Xi(&\encodern,\decodern) =\sum_{\xn, \yn, \zn} \Big| \sum_{m \in [\theta] \cup \set{0}} P_{XZ}^n(\xn,\zn) P_{M|\Xn}(m|\xn)P_{\Yn|M\Zn}(\yn|m,\zn) \\
    & \hspace{60pt}-\sum_{\sfxn}\sum_{\substack{m \in [\theta] \cup \set{0} }}  P^n_{XZ}(\sfxn,\zn)P_{M|\Xn}(m|\sfxn) P_{\Yn|M\Zn}(\yn|m,\zn) W^n_{X|YZ}(\xn
    |\yn,\zn)\Big| \nonumber \\
    &=\sum_{\substack{\xn \notin \Tx \\ \yn,\ \zn}}
    \Big|P_{XZ}^n(\xn,\zn) P^n_{Y|UZ}(\yn|u_0^n,\zn) \\
    & \hspace{50pt}
    - \sum_{\sfxn}\sum_{\substack{m \in [\theta] \cup \set{0} }}   P^n_{XZ}(\sfxn,\zn)P_{M|\Xn}(m|\sfxn) P_{\Yn|M\Zn}(\yn|m,\zn) W^n_{X|YZ}(\xn
    |\yn,\zn)\Big|\\
     &+\sum_{\substack{\xn \in \Tx \\ \yn,\ \zn}}
    \Big|P_{XZ}^n(\xn,\zn) \sum_{m\in [\theta]}P_{M|\Xn}(m|\xn) P_{\Yn|m\Zn}(\yn|m,\zn) \\
    & \hspace{50pt}
    + P_{XZ}^n(\xn,\zn)P_{M|\Xn}(0|\xn)P^n_{Y|UZ}(\yn|u_0^n,\zn)\\
    & \hspace{50pt}
    - \sum_{\sfxn}\sum_{\substack{m \in [\theta] \cup \set{0}}}  P^n_{XZ}(\sfxn,\zn)P_{M|\Xn}(m|\sfxn) P_{\Yn|M\Zn}(\yn|m,\zn) W^n_{X|YZ}(\xn
    |\yn,\zn)
    \Big|\\
    & \overset{a}{\leq} \cpe + \notce +  \er + 3 \!\!\!\!\!\!\sum_{\xn \notin \Tx}P^n_X{(\xn)} \\
    & \hspace{50pt} + \sum_{\substack{\sfxn \in \Tx \\ \xn,\ \yn,\ \zn}} P^n_{XZ}(\sfxn,\zn)P_{M|\Xn}(0|\sfxn) P_{\Yn|\Un\Zn}(\yn|u_0^n,\zn) W^n_{X|YZ}(\xn
    |\yn,\zn)\\
    & \overset{b}{\leq} \cpe + 2\notce  + \er + 4 \!\!\!\!\!\!\sum_{\xn \notin \Tx}P^n_X{(\xn)} 
    \overset{c}{\leq}\cpe + 2\notce + \er +  4\epsilon,
\end{align*}
for all sufficiently large $n$ and all $\delta>0$, where $(a)$ follows from triangle inequality and by defining terms $\cpe ,\notce,\eqand \er$ as follows:
\begin{align*}
    \cpe &\deq \sum_{\substack{\xn \in \Tx \\ \yn, \zn}}
    \Big|P_{XZ}^n(\xn,\zn) \sum_{m\in [\theta]}P_{M|\Xn}(m|\xn) P^n_{Y|UZ}(\yn|\calD(m,\zn),\zn) \\
    & \hspace{50pt} - \sum_{\substack{m \in [\theta] \\ \sfxn \in \Tx}}  P^n_{XZ}(\sfxn,\zn)P_{M|\Xn}(m|\sfxn) P^n_{Y|UZ}(\yn|\calD(m,\zn),\zn) W^n_{X|YZ}(\xn
    |\yn,\zn)
    \Big|,\\
    \notce &\deq \sum_{\xn \in \Tx} P^n_{X}(\xn) \Big(1 - \sum_{l\in[\btheta]}E_{l|\Xn}(l|\xn)\Big),  \\
    \er &\deq \sum_{\substack{\sfxn \in \Tx \\ \xn \notin \Tx}} \sum_{\substack{m \in [\theta] \\\yn \\ \zn }}
     P^n_{XZ}(\sfxn,\zn)P_{M|\Xn}(m|\sfxn) P^n_{Y|UZ}(\yn|\calD(m,\zn),\zn) W^n_{X|YZ}(\xn
    |\yn,\zn),
\end{align*}
and $(b)$ follows by writing $P_M(0) = \sum_{\sfxn \notin \Tx} P^n_X(\sfxn) + \sum_{\xn \in \Tx} P^n_X(\xn)E_{l|\Xn}(0|\xn),$ and $(c)$ follows from the standard typically argument for all sufficiently large $n$. 

\noindent \textbf{Step 1: Bounding the error induced by not covering}

\noindent Note that the error term $\notce$ captures the error induced by not covering the $n$-product side-information assisted posterior channel. We bound this term by utilizing the following proposition.
\begin{proposition} For all $\epsilon \in (0,1)$, for all sufficiently small $\eta,\delta >0$, and for all sufficiently large $n$, we have $\EE[\Ipmf \notce] \leq \epsilon$, if $\Rbar > I(X:U)$.
\end{proposition}
\begin{proof}
    The proof follows from the \cite[Proposition 8]{atif2023lossy}.
\end{proof}
Next, we move on to isolating the error component of $\cpe$ caused by binning (packing). 

\noindent \textbf{Step 2: Isolating the term induced by binning}

\noindent We consider the term corresponding to $\cpe$. By adding and subtracting an appropriate term inside the modulus of $\cpe$ and using triangle inequality, we get $\cpe \leq \ce + \peone + \petwo$, where 
\begin{align*}
    \ce &\deq \!\!\!\sum_{\substack{\xn \in \Tx \\ \un \in \Tu }} \sum_{\substack{\yn \\ \zn}} \sum_{\substack{ l \in [\btheta]\\m\in [\theta] }} 
    \frac{1}{\btheta} \frac{(1-\varepsilon)}{(1+\eta)} 
    \I_{\set{\Un(l) = \un}} \I_{\set{\calB(l) = m}}
 \Big|P_{Z|X}^n(\zn|\xn) P_{X|U}^n(\xn|\un) P^n_{Y|UZ}(\yn|\un,\zn) \\
    & \hspace{25pt} \times \I_{\set{\xn \in \Txcond}} - \!\!\!\!\!\!\!\sum_{\substack{\sfxn \in \Txcond}} \!\!\!\!\!\!P_{Z|X}^n(\zn|\sfxn) P_{X|U}^n(\sfxn|\un) P^n_{Y|UZ}(\yn|\un,\zn) W^n_{X|YZ}(\xn
    |\yn,\zn)
    \Big|,\\
    \peone &\deq \sum_{\substack{\xn \in \Tx \\ \un \in \Tu }} \sum_{\substack{\yn \\ \zn}} \sum_{\substack{ l \in [\btheta]\\m\in [\theta] }} 
    \frac{1}{\btheta} \frac{(1-\varepsilon)}{(1+\eta)} 
    \I_{\set{\Un(l) = \un}} \I_{\set{\calB(l) = m}} \I_{\set{\xn \in \Txcond}}
    P_{Z|X}^n(\zn|\xn) P_{X|U}^n(\xn|\un) \\
    & \hspace{100pt}  \times \Big| P^n_{Y|UZ}(\yn|\un,\zn) - P^n_{Y|UZ}(\yn|\calD(m,\zn),\zn)\Big|,\\
    \petwo &\deq \sum_{\substack{\xn \in \Tx \\ \un \in \Tu \\ \sfxn \in \Tx}} \sum_{\substack{\yn \\ \zn}} \sum_{\substack{ l \in [\btheta]\\m\in [\theta] }} 
    \frac{1}{\btheta} \frac{(1-\varepsilon)}{(1+\eta)} 
    \I_{\set{\Un(l) = \un}} \I_{\set{\calB(l) = m}} 
    \I_{\set{\sfxn \in \Txcond}}
    P_{Z|X}^n(\zn|\sfxn) P_{X|U}^n(\sfxn|\un) \\
    & \hspace{100pt}  \times \Big| P^n_{Y|UZ}(\yn|\un,\zn) - P^n_{Y|UZ}(\yn|\calD(m,\zn),\zn)\Big| W^n_{X|YZ}(\xn|\yn,\zn).
\end{align*}
Here, $\peone$ and $\petwo$ captures the error induced by binning. Observe that, we can establish the bound $\petwo \leq \peone$ because $\sum_{\xn \in \Tx} W^n_{X|YZ}(\xn|\yn,\zn) \leq 1$. Consequently, we have $\cpe \leq \ce + 2\peone$. To bound the term $\peone$, we provide the following proposition.
\begin{proposition}\label{prop:binning} For all $\eta,\epsilon \in (0,1),$ for all sufficiently small $\delta >0$, and sufficiently large $n$, we have $\EE[\Ipmf \peone] \leq 3\epsilon$, if $(\Rbar-R) < I(U;Z)$.
\end{proposition}
\begin{proof}
The proof is provided in Appendix \ref{app:proof:prop:binning}.
\end{proof} 
\noindent Next, we proceed to analyze the error term induced by covering. 

\noindent\textbf{Step 3: Bounding the covering error}

\noindent Using the Markov chain $U-X-Z$, $X-(U,Z)-Y$, and $X-(Y,Z)-U$ which $P_{UXYZ}$ satisfies, we can rewrite the term $ P_{X|U} P_{Z|UX} P_{Y|UXZ}$ as follows:
\begin{equation}\label{eqn:mc_prob}
    P_{Z|X} P_{X|U} P_{Y|UZ} = P_{Z|UX} P_{X|U} P_{Y|UXZ} = P_{YZ|U} W_{X|YZ}.
\end{equation}
Using \eqref{eqn:mc_prob}, we can simplify the terms inside the modulus of $\ce$ as:
\begin{align*}
    \ce &\deq \sum_{\substack{\xn \in \Tx \\ \un \in \Tu }} \sum_{\substack{\yn \\ \zn}} \sum_{\substack{ l \in [\btheta]\\m\in [\theta] }} 
    \frac{1}{\btheta} \frac{(1-\varepsilon)}{(1+\eta)} 
    \I_{\set{\Un(l) = \un}} \I_{\set{\calB(l) = m}} P_{YZ|U}^n(\yn,\zn|\un) W^n_{X|YZ}(\xn|\yn,\zn) \\
    & \hspace{2in} \times \Big|\I_{\set{\xn \in \Txcond}} - \sum_{\substack{\sfxn \in \Txcond}} W^n_{X|YZ}(\sfxn
    |\yn,\zn)
    \Big|.
\end{align*}
To bound the above-simplified term, we provide the following proposition. 
\begin{proposition}\label{prop:covering}
    For all $\eta, \epsilon \in (0,1)$, for all sufficiently small $\delta >0$, and sufficiently large $n$, we have $\EE[\Ipmf \ce] \leq \epsilon$.
\end{proposition}
\begin{proof}
The proof is provided in Appendix \ref{app:proof:prop:covering}.
\end{proof}
Following Propositions \ref{prop:binning} and \ref{prop:covering}, for all $\epsilon \in (0,1)$, for all sufficiently large $n$ and sufficiently small $\delta, \eta> 0$, we obtain, $\EE[\Ipmf \cpe] \leq \EE[\Ipmf \ce] + 2\EE[\Ipmf \peone] \leq 7\epsilon$.  Finally, we are left with the analysis of the error term $\er$.

\noindent \textbf{Step 4: Bounding the error term $\er$}

\noindent By incorporating the addition and subtraction of a suitable term and applying the triangle inequality, we further upper bound the expression of the error term $\er$ as $\er \leq \ee + \epe,$ where
\begin{align*}
     \epe &\deq \sum_{\substack{\xn \notin \Tx \\ \un \in \Tu }} \sum_{\substack{\yn \\ \zn}} \sum_{\substack{ l \in [\btheta]\\m\in [\theta] }} 
    \frac{1}{\btheta} \frac{(1-\varepsilon)}{(1+\eta)} 
    \I_{\set{\Un(l) = \un}} \I_{\set{\calB(l) = m}} 
    \sum_{\sfxn \in \Tx}P_{Z|X}^n(\zn|\sfxn) P_{X|U}^n(\sfxn|\un)\\
    & \hspace{0.5in}  \times  \I_{\set{\sfxn \in \Txcond}} \Big| P^n_{Y|UZ}(\yn|\un,\zn) - P^n_{Y|UZ}(\yn|\calD(m,\zn),\zn)\Big| W^n_{X|YZ}(\xn|\yn,\zn)\\
    \eqand \ee &\deq\sum_{\substack{\xn \notin \Tx \\ \un \in \Tu }} \sum_{\substack{\yn \\ \zn}} \sum_{\substack{ l \in [\btheta]\\m\in [\theta] }} 
    \frac{1}{\btheta} \frac{(1-\varepsilon)}{(1+\eta)} 
    \I_{\set{\Un(l) = \un}} \I_{\set{\calB(l) = m}} 
    \sum_{\sfxn \in \Tx}P_{Z|X}^n(\zn|\sfxn) P_{X|U}^n(\sfxn|\un) \\
    & \hspace{0.8in}  \times  \I_{\set{\sfxn \in \Txcond}} P^n_{Y|UZ}(\yn|\un,\zn)W^n_{X|YZ}(\xn|\yn,\zn)
\end{align*}

\noindent Similar to the error $\petwo$, note that the error term $\epe$ can be upper bounded as $\epe \leq \peone$ because $\sum_{\xn \notin \Tx} W^n_{X|YZ}(\xn|\yn,\zn) \leq 1$. Thus, for all $\epsilon \in (0,1)$, for all sufficiently large $n$ and sufficiently small $\delta, \eta > 0$, we have $\EE[\Ipmf \epe] \leq 3\epsilon$. To bound the remaining error term $\ee$, we present a proposition below.
\noindent \begin{proposition}\label{prop:encoding_error}
    For all $\epsilon \in \set{0,1}$, for all sufficiently small $\eta,\delta >0$, and sufficiently large $n$, we have $\EE[\Ipmf \ee] \leq \epsilon$.
\end{proposition}
\begin{proof}
The proof follows from the argument similar to Proposition \ref{prop:covering}. However, for completeness, we provide the proof in Appendix \ref{app:proof:prop:encoding_error}.
\end{proof}
Eventually, using Propositions \ref{prop:binning}, \ref{prop:covering}, and \ref{prop:encoding_error}, we bound the $\EE[\encodern, \decodern].$ For all $\epsilon \in (0,1)$,
\begin{equation*}
    \EE_{\codebook}[\Xi(\encodern, \decodern)] \leq \EE_{\codebook}[\Ipmf \Xi(\encodern,\decodern) + (1-\Ipmf)] \leq 17\epsilon/2,
\end{equation*}
for all sufficiently large $n$. Since $\EE_{\codebook}[\Xi(\encodern, \decodern)] \leq 17\epsilon/2$, there must exists a code $\codebook$ such that the associated error $\Xi(\encodern, \decodern) \leq 17\epsilon/2$. This completes the achievability proof.

\section{Proof of Theorem \ref{thm:connection}}
\label{sec:proof:connection}
Given 
$P_{X^nY^n}$ is the induced $n$-letter joint distribution on the source and the reconstruction vectors. 
Then, by Lemmas \ref{lem:averageSingleletter} and \ref{lemma:sols_linearEqn_close} (stated below and proof provided in \cite{sohail2025WZ}), we have
$\lim_{n \rightarrow \infty} \| P_{X_QY_Q} - P_{Y}W_{X|Y} \|_{\normalfont \text{TV}} =0,$
for some distribution $P_{Y}$  
in $\calA(\px,W_{X|Y})$ that achieves the optimality in Theorem \ref{thm:Csourcecoding},
where $P_{X_QY_Q}=\tfrac{1}{n} \sum_{i=1}^n P_{X_iY_i}.$
Then it is well known \cite[Problem 8.3]{csiszar2011information} that if one chooses, the following single-letter bounded distortion function 
$d(x,y)= -c\log_2 \prevTC(x|\xhat)+b(x)$,
and distortion level $D={\mathbb{E}[d(X,Y)]}$, where the expectation is with respect to the distribution $P_{Y} \prevTC$,
then the backward channel $\prevTC$ attains optimality in the standard rate-distortion function of the source 
$(\px,\sfX,\sfXhat,d)$
at the distortion level $D$.
Using the above arguments, we 
infer that the same protocol also achieves the 
distortion value of $D$, i.e., $\lim_{n \rightarrow \infty} \!\mathbb{E} [\tfrac{1}{n} \!\sum_{i=1}^n d(X_i,{Y_i})] \!=\!\lim_{n \rightarrow \infty} \sum_{x,y} \!P_{X_QY_Q}\!(x,y) d(x,y)\! =\!
D.$
In summary, the protocol also achieves the optimal rate-distortion function at the distortion level $D$. 
This is along anticipated lines as the current formulation uses a stricter global error criterion than the standard approach that uses a local one.

\begin{lemma} \cite[Lemma 8]{atif2023lossy}\label{lem:averageSingleletter}
The distributions $P_{X^nY^n}$ and $P_{Y^n}\prevTC^n$ satisfy
$$ \|P_X - \sum_{\xhat} P_{Y_Q}(\xhat)\prevTC(\cdot|\xhat) \|_{\normalfont \text{TV}} \leq 
\|P_{X_QY_Q} - P_{Y_Q}\prevTC \|_{\normalfont \text{TV}}  \leq 
\|P_{X^nY^n} -
    P_{Y^n}\prevTC^n\|_{\normalfont \text{TV}}.$$
\end{lemma}
\begin{lemma} \label{lemma:sols_linearEqn_close}
Let \( \sfA \in \mathbb{R}^{\alpha \times \beta} \) be a matrix and \( \sfb \in \mathbb{R}^{\alpha} \) a vector. Suppose \( x_0 \in \mathbb{R}^{\beta}_{\geq 0} \) is a non-negative vector satisfying the conditions \( \|\sfA x_0 - \sfb\|_1 \leq \delta \) and \( \I^\texttt{T} x_0 = 1 \) for some \( \delta > 0 \), where \( \I^\texttt{T} \) denotes a row vector of ones.
Then, there exists a \( x \in \mathbb{R}^{\beta}_{\geq 0} \) such that \( A x = b \) and \( \I^\top x = 1 \), with the additional property that
$\|x - x_0\|_1 \leq f(\delta),$
where \( f(\delta) \to 0 \) as \( \delta \to 0 \).
\end{lemma}
\begin{proof}
    The proof is provided in Appendix \ref{app:proof:lemma:sols_linearEqn_close}.
\end{proof}

\bibliographystyle{ieeetr}
\bibliography{reference}

\appendix

\subsection{Useful Lemmas}
\label{app:useful lemmas}
\begin{lemma}\label{app:lemma:traceinequality}
    Let $\rho$ and $\sigma$ be positive operators and $\Lambda$ a positive operator such that $0\leq \Lambda\leq I$. Then the following inequality holds
    $\Tr\{\Lambda \rho\} \leq \Tr\{\Lambda \sigma\}+\|\rho-\sigma\|_1.$
\end{lemma}
\begin{lemma}[Non-Commutative Union Bound \cite{sen2012achieving}]
   Let $\sigma$ be a sub-normalized state such that $\sigma\geq 0$ and $\Tr\{\sigma\}\leq 1$. Let $\Pi_1,\Pi_2, \Pi_N$ be projectors. Then the following holds 
   $$\Tr\{\sigma\} - \Tr\{\Pi_N\cdots\Pi_1\sigma\Pi_1\cdots\Pi_N\} \leq 2\sqrt{\sum_{i=1}^N\Tr\{(I-\Pi_i)\sigma\}}.$$
\end{lemma}

\subsection{Proof of Proposition \ref{prop:binning}}
\label{app:proof:prop:binning}
Before commencing the proof, it is noteworthy that if we assume $\calD(m,\zn) = \utilde^n$. Then, using upper bound on  total variation, we obtain the following expression:
\begin{equation}\label{eqn:packingbound}
    \sum_{\yn} \Big| P^n_{Y|UZ}(\yn|\un,\zn) - P^n_{Y|UZ}(\yn|\Tilde{u}^n,\zn)\Big| \leq 2 \cdot \I_{\set{\un \neq \Tilde{u}^n}}.
\end{equation}
Therefore, by utilizing \eqref{eqn:packingbound} and applying the union bound, we can rewrite the term $\peone$ as $\peone \leq \peoneone + \peonetwo$, where 
\begin{align*}
    \peoneone &\deq 2 \sum_{\substack{\xn \in \Tx \\ \un \in \Tu }} \sum_{\substack{ \zn \\l \in [\btheta]\\m\in [\theta] }} 
    \frac{1}{\btheta} \frac{(1-\varepsilon)}{(1+\eta)} 
    \I_{\set{\Un(l) = \un}} \I_{\set{\calB(l) = m}} \I_{\set{\xn \in \Txcond}}
    P_{Z|X}^n(\zn|\xn) P_{X|U}^n(\xn|\un) \\
    &\hspace{100pt} \times\I_{\set{(\un,\zn) \notin \Tuz}} \eqand\\
    \peonetwo &\deq 2 \sum_{\substack{\xn \in \Tx \\ \un \in \Tu }}\sum_{\substack{ \zn \\ l \in [\btheta]\\m\in [\theta]}} 
    \frac{1}{\btheta} \frac{(1-\varepsilon)}{(1+\eta)} 
    \I_{\set{\Un(l) = \un}} \I_{\set{\calB(l) = m}} \I_{\set{\xn \in \Txcond}}
    P_{Z|X}^n(\zn|\xn) P_{X|U}^n(\xn|\un) \\
    &\hspace{100pt} \times \sum_{\ltilde \in [\btheta]} \sum_{\utilde^n \neq \un} \I_{\set{\Un(\ltilde) = \utilde^n}} \I_{\set{\calB(\ltilde) = m}} \I_{\set{(\utilde^n,\zn) \in \Tuz}}.
\end{align*}
We first show that the term $\peoneone$ can be made arbitrary small for sufficiently large $n$. Consider the following inequalities:
\begin{align*}
    \EE[&\Ipmf \peoneone] \\
    &\leq \sum_{\substack{\un \in \Tu  \\\xn \in \Txcond  }}  \sum_{\substack{ \zn \\l \in [\btheta]\\m\in [\theta] }} 
    \frac{2}{\btheta} \frac{P^n_U(\un)}{(1+\eta)} 
    \I_{\set{\calB(l) = m}} 
    P_{Z|X}^n(\zn|\xn) P_{X|U}^n(\xn|\un) \I_{\set{(\un,\zn) \notin \Tuz}} \\
    &\overset{a}{=} \sum_{\substack{\un \in \Tu  \\\xn \in \Txcond  }}
   \sum_{\substack{ \zn \\ l \in [\btheta]}}
    \frac{2}{\btheta} \frac{P^n_U(\un)}{(1+\eta)}  
    P_{XZ|U}^n(\xn,\zn|\un)  \I_{\set{(\un,\zn) \notin \Tuz}} \\
    &\overset{}{\leq} \sum_{\substack{\un \in \Tu}}
   \sum_{\substack{ \zn \\ l \in [\btheta]}}
    \frac{2}{\btheta} \frac{P_{UZ}^n(\un,\zn)}{(1+\eta)}  
      \I_{\set{(\un,\zn) \notin \Tuz}} \\
    &=  \frac{2}{(1+\eta)} \sum_{(\un,\zn) \notin \Tuz}P_{UZ}^n(\un,\zn) \overset{b}{\leq} 2\epsilon,
\end{align*}
for all sufficiently large $n$ and all $\delta>0$, where $(a)$ follows from the fact that $\sum_{m \in [\theta]} \I_{\set{\calB(l) = m}} \!= \!1$ and using the Markov chain $U-X-Z$ and $(b)$ follows from the standard joint typicality argument. We now proceed to bound the error term $\peonetwo$. Consider the following set of inequalities:
\begin{align*}
    \EE[&\Ipmf \peonetwo] \\
    &\overset{a}{\leq} \sum_{\substack{\un \in \Tu \\ \xn \in \Txcond}}\sum_{\substack{ \zn \\ l \in [\btheta]\\m\in [\theta]}} 
    \frac{2}{\btheta} \frac{(1-\varepsilon)}{(1+\eta)} 
     \I_{\set{(\utilde^n,\zn) \in \Tuz}} P_{XZ|U}^n(\xn,\zn|\un) \\
    &\hspace{80pt} \times \sum_{\ltilde \neq l} \sum_{\utilde^n \neq \un} \EE[ \I_{\set{\calB(l) = m}}\I_{\set{\calB(\ltilde) = m}} \I_{\set{\Un(l) = \un}} \I_{\set{\Un(\ltilde) = \utilde^n}}] \\
    &\overset{}{=} \sum_{\substack{\un \in \Tu \\ \xn \in \Txcond}}\sum_{\substack{ \zn \\ l \in [\btheta]\\m\in [\theta]}} 
    \frac{2}{\btheta} \frac{(1-\varepsilon)}{(1+\eta)} 
     \I_{\set{(\utilde^n,\zn) \in \Tuz}} P_{XZ|U}^n(\xn,\zn|\un) \\
    &\hspace{80pt} \times \sum_{\ltilde \neq l} \sum_{\utilde^n \neq \un} \text{Pr}\{\set{\calB(l) = m, \calB(\ltilde) = m} \}\frac{P^n_U(\un)P^n_U(\utilde^n)}{(1-\varepsilon)^2} \\
    &\overset{}{\leq} \sum_{\substack{\un \in \Tu \\ \zn, \utilde^n \neq \un}}\sum_{\substack{l, \ltilde \\m\in [\theta]}} 
    \frac{2}{\btheta} \frac{(1-\varepsilon)}{(1+\eta)} \frac{1}{\theta^2}
     \I_{\set{(\utilde^n,\zn) \in \Tuz}} P_{UZ}^n(\un,\zn) \frac{P^n_U(\utilde^n)}{(1-\varepsilon)^2} \\
    &\overset{}{\leq} \sum_{\substack{\ltilde \neq l}} \frac{1}{\theta}
    \frac{2}{\btheta} \frac{(1-\varepsilon)}{(1+\eta)} 
\sum_{\substack{(\utilde^n,\zn)\in \Tuz}} \frac{1}{(1-\varepsilon)} P_{Z}^n(\zn) P^n_U(\utilde^n)\\
    &\overset{b}{\leq} \sum_{\substack{\ltilde \neq l }} 
    2 \frac{(1-\varepsilon)}{(1+\eta)} 
    \frac{(\btheta-1)}{\theta} \frac{2^{-n(I(U;Z) - \delta_1(\delta))}}{(1-\varepsilon)} \\
     &\leq \frac{1}{(1+\eta)}2^{n(\Rbar - R - I(U;Z) + \delta_1(\delta) + 1/n)}. 
\end{align*}
where $(a)$ follows by using the Markov chain $U-X-Z$, $(b)$ follows from the properties of the joint typical sequences and $\delta_1(\delta)$ is a function that follows from the characterization of the size of the typical set \cite{cover2006elements}. Therefore, if $(\Rbar -R) < I(U;Z) - \delta_1 - 1/n$, the expected error $\EE[\Ipmf \peonetwo]$ decays exponentially and can be made arbitrarily small, for sufficiently large $n$. Hence, we get $\EE[\Ipmf \peone] \leq \EE[\Ipmf \peoneone] + \EE[\Ipmf \peonetwo] \leq 3\epsilon.$ This completes the proof of Proposition \ref{prop:binning}.

\subsection{Proof of Proposition \ref{prop:covering}}
\label{app:proof:prop:covering}
Consider the following inequalities:
    \begin{align*}
         \EE[&\Ipmf \ce] \\
         &\leq \sum_{\substack{\xn \in \Tx \\ \un \in \Tu }} \sum_{\substack{\yn \\ \zn}} \sum_{\substack{ l \in [\btheta]\\m\in [\theta] }} 
    \frac{1}{\btheta} \frac{(1-\varepsilon)}{(1+\eta)} 
    \EE[\I_{\set{\Un(l) = \un}}] \I_{\set{\calB(l) = m}} P_{YZ|U}^n(\yn,\zn|\un) \\
    & \hspace{100pt} \times \Big|\I_{\set{\xn \in \Txcond}} - \sum_{\substack{\sfxn \in \Txcond}} W^n_{X|YZ}(\sfxn
    |\yn,\zn)
    \Big|W^n_{X|YZ}(\xn|\yn,\zn)\\
    &= \sum_{\substack{\xn \in \Tx \\ \un \in \Tu }} \sum_{\substack{\yn \\ \zn}} \sum_{\substack{ l \in [\btheta]\\m\in [\theta] }} 
    \frac{1}{\btheta} \frac{1}{(1+\eta)} 
    \I_{\set{\calB(l) = m}} P_{UYZ}^n(\un,\yn,\zn) \\
    & \hspace{100pt} \times \Big|\I_{\set{\xn \in \Txcond}} - \sum_{\substack{\sfxn \in \Txcond}} W^n_{X|YZ}(\sfxn
    |\yn,\zn)
    \Big|W^n_{X|YZ}(\xn|\yn,\zn)\\
    & \overset{a}{\leq} \sum_{\substack{\un \in \Tu }} \sum_{\substack{\yn \\ \zn}} \sum_{\substack{ l \in [\btheta]\\m\in [\theta] }} 
    \frac{1}{\btheta} \frac{1}{(1+\eta)} 
    \I_{\set{\calB(l) = m}} P_{UYZ}^n(\un,\yn,\zn) \\
    &\hspace{100pt}\times 2  \!\!\!\!\! \sum_{\substack{\xn \in \Txcond \\ \sfxn \notin \Txcond}} \!\!\!\!\!W^n_{X|YZ}(\xn|\yn,\zn)W^n_{X|YZ}(\sfxn|\yn,\zn)\\
    &\overset{b}{=}\sum_{\substack{\un \in \Tu \\}} \sum_{\substack{\yn \\ \zn}} 
\frac{2}{(1+\eta)}  P_{UYZ}^n(\un,\yn,\zn)  \!\!\!\!\!\sum_{\substack{\xn \in \Txcond \\ \sfxn \notin \Txcond}}\!\!\!\!\! W^n_{X|YZ}(\xn|\yn,\zn)W^n_{X|YZ}(\sfxn|\yn,\zn)\\
&\overset{c}{\leq}\sum_{\substack{\un \in \Tu }}\sum_{\substack{\yn \\ \zn}}  
\frac{2}{(1+\eta)}  P_{UYZ}^n(\un,\yn,\zn) \!\!\!\!\! \sum_{\substack{\sfxn \notin \Txcond}}W^n_{X|YZ}(\sfxn|\yn,\zn)\\
&\overset{d}{\leq}
 2\frac{(1-\varepsilon)}{(1+\eta)} \sum_{\un \in \Tu} \frac{P_{U}^n(\un)}{(1-\varepsilon)} \sum_{\xn \notin\Txcond}W^n_{X|U}(\xn|\un)\\
 &\overset{e}{\leq}
 2\epsilon \frac{(1-\varepsilon)}{(1+\eta)} \sum_{\un \in \Tu} \frac{P_{U}^n(\un)}{(1-\varepsilon)} = 2\epsilon \frac{(1-\varepsilon)}{(1+\eta)} \leq 2\epsilon,
    \end{align*}
  for all sufficiently large $n$, and all $\eta,\delta > 0$, where $(a)$ follows by splitting the summation over $\xn \in \Tx$ as summation over $\set{\xn \in \Txcond} \cap \set{\xn \in \Tx}$ and $\set{\xn \notin \Txcond} \cap \set{\xn \in \Tx}$, $(b)$ follows from the fact that $\sum_{m \in [\theta]} \I_{\set{\calB(l) = m}} = 1$, $(c)$ follows by upper bounding the term $\sum_{\xn \in \Txcond} W^n_{X|YZ}(\xn|\yn,\zn)$ by $1$, $(d)$ follows from the Markov chain $U - (Y,Z) - X$ and writing $\sum_{\yn\zn} P_{UYZ}^n W^n_{X|UYZ} = P^n_U W^n_{X|U}$, and $(e)$ follows from the conditional typically argument. This completes the proof of Proposition \ref{prop:covering}.

  \subsection{Proof of Proposition \ref{prop:encoding_error}}\label{app:proof:prop:encoding_error}
 Consider the following set of inequalities:
\begin{align*}
    \EE[&\Ipmf \ee] \\
    &\overset{a}{\leq} \sum_{\substack{\xn \notin \Tx \\ \un \in \Tu \\ \sfxn \in \Txcond}}\sum_{\substack{\yn \\ \zn}} \sum_{\substack{ l \in [\btheta]\\m\in [\theta] }} 
    \frac{1}{\btheta} \frac{P^n_U(\un)}{(1+\eta)} 
    \I_{\set{\calB(l) = m}} 
    P_{XZ|U}^n(\sfxn,\zn|\un)P^n_{Y|UZ}(\yn|\un,\zn)W^n_{X|YZ}(\xn|\yn,\zn) \\
    &\overset{b}{\leq} \sum_{\substack{\xn \notin \Tx \\ \un \in \Tu }}\sum_{\substack{\yn\\\zn}} 
    \frac{1}{(1+\eta)} 
    P_{UYZ}^n(\un,\yn,\zn)W^n_{X|YZ}(\xn|\yn,\zn) \sum_{\sfxn \in \Txcond}W^n_{X|YZ}(\sfxn|\yn,\zn) \\
    &\overset{c}{\leq} \sum_{\substack{\xn \notin \Tx \\ \un \in \Tu }} 
    \frac{1}{(1+\eta)} 
    P_{U}^n(\un)W^n_{X|U}(\xn|\un)\\
    &\overset{}{\leq} \sum_{\substack{\xn \notin \Txcond \\ \un \in \Tu }} 
    \frac{1}{(1+\eta)} 
    P_{U}^n(\un)W^n_{X|U}(\xn|\un)\\
    &\overset{d}{\leq} 
    \epsilon\frac{(1-\varepsilon)}{(1+\eta)} 
    \sum_{\un\in \Tu}P_{U}^n(\un) =    \epsilon\frac{(1-\varepsilon)}{(1+\eta)} \leq \epsilon,
\end{align*}
where $(a)$ follows from the Markov chain $U-X-Z$, $(b)$ follows from the Markov chain $X - (U,Z) - Y$ and the fact that $\sum_{m\in [\theta]} \I_{\set{\calB(l) = m}} =1$, $(c)$ follows from the Markov chain $U- (Y,Z) - X$, and $(d)$ follows from the standard conditional typicality argument. This completes the proof of Proposition \ref{prop:encoding_error}. 

\subsection{Proof of Proposition \ref{prop:code_dependent_RV}}
\label{app:prop:code_dependent_RV}
Consider the following inequalities:
\begin{align*}
    \EE_\PP[E_1] &= \sum_{m\in [\Theta]}\sum_{k\in [\bTheta]} {(\Theta\bTheta)}^{-1}\EE_\PP[\Tr\{\rhotilde_{m,k}^{\refe B}\}]\\
    & = \sum_{m\in [\Theta]}\sum_{k\in [\bTheta]} {(\Theta\bTheta)}^{-1}\!\!\!\!\!\!\sum_{\xn \in \Txqc}\!\!\!\frac{P^n_X(\xn)}{(1-\varepsilon)}\Tr\{\rhotilde_{\xn}^{\refe B}\}\\
    & \geq \!\!\!\!\!\!\sum_{\xn \in \Txqc}\!\!\!\frac{P^n_X(\xn)}{(1-\varepsilon)}(1-2\varepsilon-2\sqrt{\varepsilon}) \geq (1-\epsilon),
\end{align*}
where the first inequality follows from \cite[Eqn. 23]{wilde_e}. Now, using \eqref{eq:closeness_ref_SI}, we get 
\begin{align*}
    \EE_\PP[E_2] &= \sum_{m\in [\Theta]}\sum_{k\in [\bTheta]}{(\Theta\bTheta)}^{-1}\ \EE_\PP[\|\rhotilde_{m,k}^{\refe B} - \calW^{\refe B}_{m,k}\|_1]\\
    &= \sum_{m\in [\Theta]}\sum_{k\in [\bTheta]} {(\Theta\bTheta)}^{-1}\!\!\!\!\!\!\sum_{\xn \in \Txqc}\!\!\!\frac{P^n_X(\xn)}{(1-\varepsilon)}\|\rhotilde_{\xn}^{\refe B} - \calW^{\refe B}_{\xn}\|_1\leq \epsilon.
\end{align*}
\subsection{Proof of Proposition \ref{prop:qc_NC}}
\label{app:prop:proof:qc_NC}
    From equations \eqref{eqn:omegamk} and \eqref{eqn:lambdamk}, we rewrite $ \zeta_{\text{NC}}$ as
    \begin{align*}
        \zeta_{\text{NC}} &
        = \sum_{\substack{k'\in [\bTheta]\cup \{0\}}}\!\!\!\! \Tr\big\{\big({I^{\refe}}^{\tensor n} \tensor \Lambda_{k'}^{(m)}\big)\omega^{RB}_{0,0}\big\} = \Tr\big\{\big({I^{\refe}}^{\tensor n} \tensor \!\!\!\!\!\!\sum_{\substack{k'\in [\bTheta]\cup \{0\}}}\!\!\!\!\!\!\dpovm_{k'}^{(m)} \big)\omega^{RB}_{0,0}\big\} = \Tr\{\omega^{RB}_{0,0}\},
    \end{align*}
    where the equality follows because $\sum_{\substack{k'\in [\bTheta]\cup \{0\}}}\dpovm_{k'}^{(m)} = I$.
    Consider the following inequalities:
    \begin{align*}
        \EE_\PP[\I_{\curly{\mbox{\normalfont sP}}}\zeta_{\text{NC}}] \overset{a}{\leq} \EE_\PP[\Tr\{\omega^{RB}_{0,0}\}] &= 
        \EE_\PP\Big[\Tr\Big\{\Big(I-\sum_{m\in [\Theta]}\sum_{k\in[\bTheta]}\epovm_{m,k}^{\refe B}\Big) {\calW^{\refe B}}^{\tensor n}\Big\}\Big]\\
        &=
        1-\sum_{m\in [\Theta]}\sum_{k\in[\bTheta]}\EE_\PP\Big[\Tr\Big\{\epovm_{m,k}^{\refe B}{\calW^{\refe B}}^{\tensor n}\Big\}\Big]\\
        &= 1-\frac{(1-\varepsilon)}{(1+\eta)}\sum_{m\in [\Theta]}\sum_{k\in[\bTheta]}(\Theta\bTheta)^{-1}\EE_{\PP}[\Tr\{\rhotilde_{m,k}^{\refe B}\}]\\
        &\overset{b}{\leq}1-\frac{(1-\varepsilon)}{(1+\eta)}(1-\epsilon) \leq \tilde{\epsilon} \in (0,1),
    \end{align*}
    for all sufficiently large $n$ and all sufficiently small $\eta, \delta>0$, where $(a)$ is based on the fact that $\I_{\curly{\mbox{\normalfont sP}}}\leq 1$ and  $(b)$ follows from Proposition \ref{prop:code_dependent_RV}. This completes the proof of Proposition \ref{prop:qc_NC}. 

\subsection{Proof of Proposition \ref{prop:qc_NP}}
\label{app:prop:proof:qc_NP}
Consider the following inequalities:
\begin{align*}
    \EE_\PP[\I_{\curly{\mbox{\normalfont sP}}}\zeta_{\text{NP}}] &\leq \EE_\PP\Big[\sum_{m \in [\Theta]}\sum_{\substack{k\in [\bTheta]}}
\Tr\big\{\big(I-\dpovm_{k}^{(m)} \big)\omega^{B}_{m,k}\big\}\}\Big]\\
&\overset{a}{\leq} \EE_\PP\bigg[\sum_{m \in [\Theta]}\sum_{\substack{k\in [\bTheta]}}\gamma\ \Tr\big\{\big(I-\dpovm_{k}^{(m)} \big)\calW^{B}_{m,k}\big\}\} + \sum_{m \in [\Theta]}\sum_{\substack{k\in [\bTheta]}}\gamma\ \|\rhotilde^B_{m,k}-\calW^B_{m,k}\|_1\bigg]\\
&\overset{b}{\leq} \frac{(1-\varepsilon)}{(1+\eta)}\ \EE_\PP\bigg[\frac{1}{\Theta}\sum_{m \in [\Theta]}\frac{1}{\bTheta}\sum_{\substack{k\in [\bTheta]}}
\Tr\big\{\big(I-\dpovm_{k}^{(m)} \big)\calW^{B}_{m,k}\big\}\}\bigg]\\
&\hspace{50pt}+ \frac{(1-\varepsilon)}{(1+\eta)}\ \EE_\PP\bigg[\sum_{m \in [\Theta]}\sum_{\substack{k\in [\bTheta]}}
(\Theta\bTheta)^{-1}\|\rhotilde^{\refe B}_{m,k}-\calW^{\refe B}_{m,k}\|_1\bigg]\overset{c}{\leq} 2 \frac{(1-\varepsilon)}{(1+\eta)}\epsilon \leq 2\epsilon,
\end{align*}
where $(a)$ follows from the trace inequality: $\Tr\{\Lambda \rho\} \leq \Tr\{\Lambda \sigma\}+\|\rho-\sigma\|_1$ (Lemma \ref{app:lemma:traceinequality} in Appendix \ref{app:useful lemmas}), $(b)$ follows from the definition of  $\gamma = \frac{(1-\varepsilon)}{(1+\eta)}(\Theta\bTheta)^{-1}$ and monotonicity of trace distance, and $(c)$ follows from Proposition \ref{prop:code_dependent_RV} and a weaker version of Proposition \ref{prop:qc_packing}, i.e., 
\begin{equation}\label{eqn:qc_packing}
    \EE_\PP\left[\frac{1}{\Theta}\sum_{m\in [\Theta]}\frac{1}{\bTheta}\sum_{k\in[\bTheta]} \Tr\left\{(I-\dpovm_{k}^{(m)})\calW_{m,k}^B\right\}\right] \leq \epsilon,
\end{equation}
for sufficiently small $\delta>0$ and for all sufficiently large $n$, if $\frac{1}{n}\log(\bTheta)<\chi(\{P_X(x),\calW^B_x\}).$
This completes the proof of Proposition \ref{prop:qc_NP}. 
\subsection{Proof of Proposition \ref{prop:qc_cov}}
\label{app:prop:proof:qc_cov}
Consider the following inequalities:
    \begin{align*}
\EE_\PP[&\I_{\curly{\mbox{\normalfont sP}}}\zeta_{\text{C}}] \\
&\leq\EE_\PP\Big[\sum_{m \in [\Theta]}\sum_{\substack{k\in [\bTheta]}}\Tr\{\lambda_{m,k}^{\refe B}\}\|\omegabar_{m,k}^{\refe B} -  \calW_{m,k}^{\refe B}\|_1\Big]\\
&\overset{a}{\leq}\EE_\PP\Big[\sum_{m \in [\Theta]}\sum_{\substack{k\in [\bTheta]}}\Tr\{\omega_{m,k}^{\refe B}\}\|\omegabar_{m,k}^{\refe B} -  \calW_{m,k}^{\refe B}\|_1\Big]\\
&\overset{b}{\leq}\EE_\PP\Big[\sum_{m \in [\Theta]}\sum_{\substack{k\in [\bTheta]}}\gamma \ \Tr\{\rhotilde_{m,k}^{\refe B}\}\Big\|\frac{\rhotilde_{m,k}^{\refe B}}{\Tr\{\rhotilde_{m,k}^{\refe B}\}} -  \rhotilde_{m,k}^{\refe B}\Big\|_1 + \sum_{m \in [\Theta]}\sum_{\substack{k\in [\bTheta]}}\gamma \ \Tr\{\rhotilde_{m,k}^{\refe B}\}\|\rhotilde_{m,k}^{\refe B} -  \calW_{m,k}^{\refe B}\|_1\Big]\\
&\overset{c}{\leq}\sum_{m \in [\Theta]}\sum_{\substack{k\in [\bTheta]}}\gamma \ (1-\EE_\PP[\Tr\{\rhotilde_{m,k}^{\refe B}\}]) + \sum_{m \in [\Theta]}\sum_{\substack{k\in [\bTheta]}}\gamma \  \EE_\PP[\|\omega_{m,k}^{\refe B} -  \calW_{m,k}^{\refe B}\|_1]\\
&\overset{}{\leq}\frac{(1-\varepsilon)}{(1+\eta)}\Big(1-\sum_{m \in [\Theta]}\sum_{\substack{k\in [\bTheta]}}(\Theta\bTheta)^{-1}\EE_\PP[\Tr\{\rhotilde_{m,k}^{\refe B}\}] + \sum_{m \in [\Theta]}\sum_{\substack{k\in [\bTheta]}}(\Theta\bTheta)^{-1} \EE_\PP[\|\omega_{m,k}^{\refe B} -  \calW_{m,k}^{\refe B}\|_1]\Big)\\
&\overset{d}{\leq}2\frac{(1-\varepsilon)}{(1+\eta)}\epsilon \leq 2\epsilon,
    \end{align*}
    where $(a)$ follows from the application of the trace inequality: $\Tr{\Lambda \rho} \leq \Tr{\rho}$ for $0\Lambda\leq I$ and positive operator $\rho$, $(b)$ follows by adding and subtracting appropriate terms and then applying triangle inequality, $(c)$ follows from the definition of $\gamma$, and $(d)$ follows from Proposition \ref{prop:code_dependent_RV}. This completes the proof of Proposition \ref{prop:qc_cov}.

  \subsection{Proof of Lemma \ref{lemma:sols_linearEqn_close}}\label{app:proof:lemma:sols_linearEqn_close}
We begin the proof by considering the following convex optimization problem:
\begin{align}
    \min_{x} & \frac{1}{2}\|x-x_0\|_2^2 \nonumber \\
    \text{subject to } & \sfA x=\mathsf{b}, \I^\texttt{T}x = 1,x \geq 0. 
\end{align}
Our objective is first to find the unique optimal solution $x^*$ of the above convex optimization problem and then show that for any $x_0 \geq 0$, if $\|\sfA x_0-\mathsf{b}\|_1 \leq \delta$ and $\I^\texttt{T}x_0 = 1$, then  $x^*$ satisfies $\|x^* -x_0\|_1 \leq \tilde{\delta}(\delta)$. 
Toward solving the optimization problem, let, ${\mathbf{A}}_{} := [\sfA^\texttt{T} \; \I]^\texttt{T}$  and ${\mathbf{b}} := [\mathsf{b}^\texttt{T}\; 1]^\texttt{T}$. Assume $\bfb \in \text{col}(\bfA)$, i.e., $\bfb$ belongs to the column space of $\bfA$, otherwise, $\bfA x=\bfb$ has no solution. Moreover, without loss of generality, assume $\bfA$ is full row rank; if not, one can remove the linearly dependent rows to make it full row rank.

The Lagrangian for this problem is 
$$\mathcal{L}(x,\lambda,\mu):= \frac{1}{2}\|x-x_0\|^2_2 + \lambda^\texttt{T}(\mathbf{A}x_0-\mathbf{b}) - \mu^\texttt{T}x,$$
where $\lambda$ and $\mu\geq 0$ are the Lagrange multipliers for the equality constant $\mathbf{A}x=\mathbf{b}$ and non-negativity constraint $x\geq 0$, respectively. The KKT conditions are as follows:
\begin{itemize}
    \item \textit{Stationarity}:  The gradient of the Lagrangian with respect to 
$x$ must vanish, i.e.,
$$\partial_x \mathcal{L}(x,\lambda,\mu) = (x-x_0) + \mathbf{A}^\texttt{T}\lambda-\mu = 0.$$
This simplifies to $x = x_0-\mathbf{A}^\texttt{T}\lambda+\mu.$
\item \textit{Primal Feasibility}: $\mathbf{A}x=\mathbf{b}$. Substituting the value of $x$ from the stationarity condition, we get $(\bfA\bfA^\textT)\lambda = (\bfA x_0-\bfb)+\bfA\mu$.  Thus, we get, $$\lambda = (\bfA\bfA^\textT)^{-1}(\bfA x_0-\bfb)+(\bfA\bfA^\textT)^{-1}\bfA\mu.$$
Once the $\lambda$ is determined, substituting it back into stationary condition, we get $$x = x_0 - \bfA^\textT(\bfA\bfA^\textT)^{-1}(\bfA x_0-\bfb)+(I-\bfA^\textT(\bfA\bfA^\textT)^{-1}\bfA)\mu.$$
\item \textit{Dual Feasibility}: $x\geq 0 \eqand \mu \geq 0$.
\item \textit{Complementary Slackness}: $\mu_i x_i= 0$ for all $i \in [\beta].$
\end{itemize}
We can re-write the above expression of $x$ using the generalized inverse\footnote{Often known as pseudoinverse or Moore–Penrose inverse.}\cite{ben2006generalized} as $$x = x_0 - \bfA^{+}(\bfA x_0-\bfb)+(I-\bfA^{+}\bfA)\mu,$$
where $\bfA^+ = \bfA^\textT(\bfA\bfA^\textT)^{-1}$ when $\bfA$ is full-row rank. The above expression is also valid for arbitrary matrix $\bfA \in \RR^{\alpha\times\beta}$, where the generalized inverse is computed using singular value decomposition (SVD) as follows. Let \(\mathbf{A} = U_{\alpha \times \alpha} \Sigma_{\alpha \times \beta} V^{\textT}_{\beta \times \beta}\) be the SVD of \(\mathbf{A}\), where \(U\) and \(V\) are unitary matrices, and \(\Sigma\) is a rectangular diagonal matrix with non-negative real numbers on the diagonal. The generalized inverse of \(\mathbf{A}\) is given by \(\mathbf{A}^+ = V \Sigma^+ U^{\textT}\), where \(\Sigma^+_{\beta \times \alpha}\) is obtained by taking the reciprocal of each non-zero element on the diagonal, leaving the zeros unchanged, and finally transposing the resulting matrix.

\noindent Now, using complementary slackness condition, we find $x^*$ as follows: Let $w := x_0 - \bfA^+(\bfA x_0-\bfb)$. 
\begin{itemize}
    \item If $w\geq 0$, $x^* = w$. 
    \item If $w_{\sfS} \leq 0,$ for some $\sfS\subseteq [\beta]$. Then, $x^* = x_0 - \bar{\bfA}^+(\bar{\bfA} x_0-\bfb)$, where $\bar{\bfA}$ by setting the entries in the columns corresponding to the set 
$\sfS$ to zero for every row, except for the last row, which consists entirely of ones. 
\end{itemize}
So far, we have constructed a $x^*$, which minimizes the above-stated optimization problem. Now, we find the objective function value corresponding to $x^*$.
\begin{align*}
   \|x^*-x_0\|_1 &\overset{a}{\leq} \sqrt{\beta}\|x^*-x_0\|_2 \\
    &= \sqrt{\beta} \|\bfA^{+}(\bfA x-\bfb)\|_2 \\
    &\overset{b}{\leq} \sqrt{\beta} \|\bfA^{+}\|_2 \|(\bfA x-\bfb)\|_2 \\
    &\overset{c}{\leq} \sqrt{\beta} \|\bfA^{+}\|_2 \|(\bfA x-\mathsf{b})\|_1 \\
    &\overset{d}{\leq} \delta \sqrt{\beta} \|\bfA^{+}\|_2 \\
    &\overset{e}{=} \frac{\delta \sqrt{\beta}}{\sigma^{+}_{\min}} \quad \text{ where $\sigma^{+}_{\min}$ is the smallest non-zero singular value of $\bfA$},
\end{align*}
where $(a)$ follows from Cauchy-Schwartz inequality, $(\mathsf{b})$ follows from the sub-multiplicative property of matrix norm, $(c)$ follows from triangle inequality (or $\|\cdot\|_p \leq \|\cdot\|_1$ from Minkowski inequality), $(d)$ follows from the statement of the lemma, and $(e)$ follows from the definition of 2-norm of a matrix, i.e., the largest singular value.  Similarly, for the other case, when $w_\sfS \leq 0$ for some $S\in[\beta]$. We get $\|x^*-x_0\|_1 \leq \frac{\delta\sqrt{\beta}}{\sigmabar^{+}_{\text{min}}}$, where $\sigmabar^{+}_{\min}$ is the smallest non-zero singular value of $\bar{\bfA}$. This completes the proof of Lemma \ref{lemma:sols_linearEqn_close}.








\end{document}